\documentclass{llncs}

\usepackage{url}
\usepackage{xspace}
\usepackage{pbox}
\usepackage{array}
\usepackage{paralist}
\usepackage{algorithm}
\usepackage{algorithmic}
\usepackage{graphicx}
\usepackage{amsmath,amssymb}

\newcommand{\tran}[1]{\xrightarrow[]{#1}}

\def\ff{\mathit{ff}}
\def\tt{\mathit{tt}}
\def\f#1{\mathit{xnf}(#1)}
\def\u#1{\mathit{PA}(#1)}

\def\ue{ELIMINATION }
\def\sp{SAT\_PURSUING }
\def\up{CONFLICT\_ANALYZE }
\def\aalta{Aalta}
\def\nuxmv{NuXmv}
\def \tool {Aalta\_v2.0\xspace}

\begin{document}

\title{SAT-Based Explicit LTL Reasoning}
\author{Jianwen Li\inst{1,2} \and 
Shufang Zhu\inst{2} \and
Geguang Pu\inst{2} \and
Moshe Y. Vardi\inst{1}
}

\institute
{
  \inst{}
Department of Computer Science, Rice University, USA 
\and
  \inst{}%
Shanghai Key Laboratory of Trustworthy Computing, East China Normal University, 
 P. R.China
}

\maketitle
\begin{abstract}

We present here a new explicit reasoning framework for linear temporal logic 
(LTL), which is built on top of propositional satisfiability (SAT) solving. 
As a proof-of-concept of this framework, we describe a new LTL satisfiability algorithm. 
We implemented the algorithm in a tool, \tool, which is built on top of the Minisat SAT solver. 
We tested the effectiveness of this approach by demonstrating that \tool significantly outperforms 
all existing LTL satisfiability solvers. 

\end{abstract}

\section{Introduction}
Linear Temporal Logic (LTL) was introduced into program verification in \cite{Pnu77}. 
Since then it has been widely accepted as a language for the specification of ongoing
computations \cite{MP92} and it is a key component in the verification of reactive 
systems \cite{CGP99,Hol03}. Explicit temporal reasoning, which involves an explicit construction
of temporal transition systems, is a key algorithmic component in this context. For example, 
explicitly translating LTL formulas to B\"uchi automata is a key step both in 
explicit-state model checking \cite{GPVW95} and in runtime verification~\cite{TRV12}. 
LTL satisfiability checking, a step that should take place \emph{before} verification, 
to assure consistency of temporal requirements, also uses explicit reasoning \cite{RV10}.
These tasks are known to be quite demanding computationally for complex temporal
properties \cite{GPVW95,RV10,TRV12}. A way to get around this difficulty is to
replace explicit reasoning by symbolic reasoning, e.g., as in BDD-based or SAT-based
model checking \cite{McM93,McM03}, but in many cases the symbolic approach is
inefficient \cite{RV10} or inapplicable \cite{TRV12}. Thus, explicit temporal reasoning
remains an indispensable algorithmic tool.

The main approach to explicit temporal reasoning is based on the \emph{tableau} technique,
in which a recursive \emph{syntactic} decomposition of temporal formulas drives the construction
of temporal transition systems. This approach is based on the technique of \emph{propositional
tableau}, whose essence is search via \emph{syntactic splitting} \cite{dAgo99}. This is in contrast
to modern propositional satisfiability (SAT) solvers, whose essence is search via \emph{semantic splitting}
\cite{MZ09}.  The tableau approach to temporal reasoning underlies both the best LTL-to-automata 
translator \cite{DP04} and the best LTL-satisfiability checker \cite{LZPVH13}. Thus, we have a 
situation where in the symbolic setting much progress is being attained both by the impressive 
improvement in the capabilities of modern SAT solvers \cite{MZ09} as well as new SAT-based model-checking 
algorithms \cite{Bra11,CS12}, while progress in explicit temporal reasoning is slower and does not 
fully leverage modern SAT solving. (It should be noted that several LTL satisfiability solvers,
including {\aalta} \cite{LPZHVCoRR14} and ls4 \cite{SW12} do employ SAT solvers,
but they do so as an \emph{aid} to the main reasoning engine, rather than serve as the \emph{main}
reasoning engine.)

Our main aim in this paper is to study how SAT solving can be \emph{fully leveraged} in explicit 
temporal reasoning.  The key intuition is that explicit temporal reasoning consists of construction 
of states and transitions, subject to temporal constraints. Such temporal constraints can be reduced to
a sequence of Boolean constraints, which enables the application of SAT solving. This idea underlies
the complexity-theoretic analysis in \cite{Var89}, and has been explored in the context of modal
logic \cite{GS96}, but not yet in the context of explicit temporal reasoning.  Our belief
is that SAT solving would prove to be superior to tableau in that context.

We describe in this paper a general framework for SAT-based explicit temporal reasoning.
The crux of our approach is a construction of temporal transition system
that is based on SAT-solving rather than tableau to construct states and transitions.
The obtained transition system can be used for LTL-satisfiability solving, LTL-to-automata
translation, and runtime-monitor construction.

As proof of concept for the new framework, we use it to develop a SAT-based algorithm for 
LTL-satisfiability checking. We also propose several heuristics to speed up the checking 
by leveraging SAT solvers. We implemented the algorithm and heuristics in an LTL-satisfiability 
solver \tool. 
To evaluate its performance, we compared it against {\aalta}, the existing  
best-of-breed LTL-satisfiability solver \cite{LZPVH13,LPZHVCoRR14}, which is tableau-based. 
We also compare it against {\nuxmv}, a symbolic LTL-satisfiability solver that is based on 
cutting-edge SAT-based model-checking algorithms \cite{Bra11,CS12}, which outperforms {\aalta}.
We show that our explicit SAT-based LTL-satisfiability solver outperforms both.

In summary, the contributions in this paper are as follows:
\begin{itemize}
\item We propose a SAT-based explicit LTL-reasoning framework.
\item We show a successful application of the framework to LTL-satisfiability checking, 
by designing a novel algorithm and efficient heuristics.
\item We compare our new framework for LTL-satisfiability checking with existing approaches. 
The experimental results demonstrate that our tool significantly outperforms other existing 
LTL satisfiability solvers.   
\end{itemize}

The paper is organized as follows. Section \ref{sec:pre} provides technical background.
Section \ref{sec:framework} introduces the new SAT-based explicit-reasoning 
framework.  Section \ref{sec:ltlsat} describes in detail the application to LTL-satisfiability checking. 
Section \ref{sec:experiment} shows the experimental results for LTL-satisfiability checking. 
Finally Section \ref{sec:discussion} provides concluding remarks. Missing proofs are in the Appendix. 

\section{Preliminaries}\label{sec:pre}

Linear Temporal Logic (LTL) is considered as an extension of propositional logic, 
in which temporal connectives $X$ (next) and $U$ (until) are introduced. 
Let $AP$ be a set of atomic properties.  The syntax of LTL formulas is defined by:
\begin{align*}
\phi\ ::=\ \tt \mid \ff\mid a \mid \neg \phi \ |\ \phi\wedge \phi\ |\ \phi\vee \phi\ |\ \phi U \phi\ |\ X \phi
\end{align*}
where $a\in AP$, $\tt$ is \textit{true} and $\ff$ is \textit{false}. We introduce the $R$ (release) connectives 
as the dual of $U$, which means $\phi R\psi \equiv \neg (\neg\phi U\neg\psi)$. We also use the usual
abbreviations: $Fa=\tt Ua$, and $Ga=\ff Ra$. 

We say that $a$ is a \emph{literal} if it is an atomic proposition or its negation. 
Throughout the paper, we use $L$ to denote the set of literals, lower case letters 
$a,b,c,l$ to denote literals, $\alpha$ to denote propositional formulas, and 
$\phi, \psi$ for LTL formulas.  We consider LTL formulas in negation normal form (NNF), 
which can be achieved by pushing all negations in front of only atoms. 
Since we consider LTL in NNF, formulas are interpreted here on infinite literal sequences, 
whose alphabet is $\Sigma:= 2^L$. 

A \emph{trace} $\xi=\omega_0\omega_1\omega_2\ldots$ is an infinite sequence in
$\Sigma^\omega$. For $\xi$ and $k\geq 1$ we use $\xi ^k=\omega_0\omega_1\ldots
\omega_{k-1}$ to denote a prefix of $\xi$,
and $\xi _k =\omega_k\omega_{k+1}\ldots$ to denote a suffix of $\xi$.
Thus, $\xi=\xi^k \xi_k$. 
The semantics of LTL with respect to an infinite trace $\xi$ is given by: 
\begin{itemize}
\item $\xi\models \alpha$ iff $\xi^1\models\alpha$, where $\alpha$ is a propositional formula;
\item $\xi \models X\ \phi$ iff $\xi_1\models \phi$;
\item $\xi \models \phi_1\ U\ \phi_2$ iff there exists $i\geq 0$
such that $\xi_i\models \phi_2$ and for all $0 \leq j < i$, $\xi_j\models \phi_1$;
\item $\xi \models \phi_1\ R\ \phi_2$ iff for all $i\geq 0$, it holds 
  $\xi_i\vDash\phi_2$ or there exists $0\leq j\leq i$ such that $\xi_j\models\phi_1$.
\end{itemize}

The \textit{closure} of an LTL formula $\phi$, denoted as $cl(\phi)$, 
is a formula set such that: (1). $\phi$ is in $cl(\phi)$; 
(2). $\psi$ is in $cl(\phi)$ if $\phi=X\psi$ or $\phi=\neg\psi$; 
(3). $\phi_1,\phi_2$ are in $cl(\phi)$ if $\phi=\phi_1\ op\ \phi_2$, 
where $op$ can be $\wedge,\vee,U$ and $R$; 
(4). $(X\psi)\in cl(\phi)$ if $\psi\in cl(\phi)$ and $\psi$ is an Until or Release formula. 
We say each $\psi$ in $cl(\phi)$, which is added via rules (1)-(3), is a \textit{subformula} 
of $\phi$.  Note that the standard definition of LTL closure consists only of rules (1)-(3).
Rule (4) is added in this paper due to its usage in later sections. 
Note that the size of $cl(\phi)$ is linear in the length of $\phi$, 
even with the addition of rule (4).

\section{Explicit LTL Reasoning}\label{sec:framework}
In this section we introduce the framework of explicit LTL reasoning.  
To demonstrate clearly both the similarity and difference between our 
approach and previous ones, we organize this section as follows. 
We first provide a general definition of temporal transition systems,
which underlies both our new approach and previous approach. We then 
discuss how traditional methods and our new one relate to this framework.

\subsection{Temporal Transition System}
As argued in \cite{Var89c,GS96}, the key to efficient \emph{modal} reasoning is to reason 
about states and transitions \emph{propositionally}. We show here how the same approach
can be applied to LTL. Unlike \emph{modal logic}, where there is a clear separation between
formulas that talk about the current state and formulas that talk about successor states
(the latter are formulas in the scope of $\Box$ or $\Diamond$, i.e. G or F in LTL), LTL formulas do not allow for 
such a clean separation. Achieving such a separation requires some additional work.

We first define propositional satisfiability of LTL formulas.
\begin{definition}[Propositional Satisfiability]\label{def:propsat}
For an LTL formula $\phi$, a \emph{propositional assignment for} $\phi$ is 
a set $A\subseteq cl(\phi)$ such that
  \begin{itemize}
    \item every literal $\ell\in L$ is either in $A$ or its negation is, but not both.
    \item $(\theta_1\wedge \theta_2)\in A$ implies $\theta_1\in A$ and $\theta_2\in A$,
    \item $(\theta_1\vee \theta_2)\in A$ implies $\theta_1\in A$ or $\theta_2\in A$,
    \item $(\theta_1 U \theta_2)\in A$ implies $\theta_2\in A$ or both $\theta_1\in A$
          and $(X(\theta_1 U \theta_2))\in A$. 
          In the former case, that is, $\theta_2\in A$, we say that $A$ satisfies 
          $(\theta_1 U \theta_2)$ \emph{immediately}. In the latter case,
          we say that $A$ \emph{postpones} $(\theta_1 U \theta_2)$. 
    \item $(\theta_1 R \theta_2)\in A$ implies $\theta_2\in A$ and either $\theta_1\in A$
          or $(X(\theta_1 R \theta_2))\in A$. 
          In the former case, that is, $\theta_1\in A$, we say that $A$ satisfies 
          $(\theta_1 R \theta_2)$ \emph{immediately}. In the latter case,
          we say that $A$ \emph{postpones} $(\theta_1 R \theta_2)$. 
  \end{itemize}
We say that a propositional assignment $A$ \emph{propositional satisfies} $\phi$, 
denoted as $A\models_p\phi$, if $\phi\in A$.  We say an LTL formula $\phi$ 
is \emph{propositionally satisfiable} if there is a propositional assignment $A$ for $\phi$
such that $A\models_p\phi$. 
\end{definition}

For example, consider the formula $\phi=(a U b) \wedge (\neg b)$. 
The set $A_1=\{a, (a U b), (\neg b),\\ (X(aUb))\}\subseteq cl(\phi)$ is a propositional assignment 
that propositionally satisfies $\phi$.  In contrast, the set $A_2=\{(a U b), \neg b\}\subseteq cl(\phi)$ 
is not a propositional assignment.

The following theorem shows the relationship between LTL formula $\phi$ and its propositional assignment. 
\begin{theorem}\label{thm:propassign}
For an LTL formula $\phi$ and an infinite trace $\xi\in \Sigma^{\omega}$, 
we have that $\xi\models\phi$ iff 
there exists a propositional assignment $A\subseteq cl(\phi)$ such that 
$A$ propositionally satisfies $\phi$ and $\xi\models \bigwedge A$.
\end{theorem}

Since a propositional assignment of LTL formula $\phi$ contains the information for both current and next states, 
we are ready to define the \emph{transition systems} of LTL formula.

\begin{definition}\label{def:rs}
Given an LTL formula $\phi$, the \emph{transition system} $T_{\phi}$ is a tuple $(S,S_0,T)$ where 
\begin{itemize}
\item 
$S$ is the set of states $s \subseteq cl(\phi)$ that are
propositional assignments for $\phi$. The \emph{trace} of a state $s$ is $s\cap L$,
that is, the set of literals in $s$.
\item
$S_0\subseteq S$ is a set of \emph{initial states}, where $\phi\in s_0$ for all $s_0\in S_0$.
\item 
$T: S\times S$ is the transition relation, where $T(s_1,s_2)$ holds if
$(X\theta)\in s_1$ implies $\theta\in s_2$, for all $X\theta\in cl(\phi)$.
\end{itemize}
A \emph{run} of $T_\phi$ is an infinite sequence $s_0,s_1,\ldots$ such that
$s_0\in S_0$ and $T(s_i,s_{i+1})$ holds for all $i\geq 0$.
\end{definition}

Every run $r=s_0,s_1,\ldots$ of $T_\phi$ induces a trace $trace(r)=trace(s_0),trace(s_1),\ldots$ in $\Sigma^\omega$.
In general, it needs not hold that $trace(r)\models\phi$. This requires an additional condition.
Consider an Until formula $(\theta_1 U \theta_2)\in s_i$. Since $s_i$ is a propositional assignment for
$\phi$ we either have that $s_i$ satisfies $(\theta_1 U \theta_2)$ immediately or that it postpones it,
and then $(\theta_1 U \theta_2)\in s_{i+1}$. If $s_j$ postpones $(\theta_1 U \theta_2)$ for all $j\geq i$,
then we say that $(\theta_1 U \theta_2)$ is \emph{stuck} in $r$.

\begin{theorem}\label{thm:reasoning}
Let $r$ be a run of $T_{\phi}$. If \emph{no} Until subformula is stuck at $r$, then $trace(r)\models\phi$.
Also, $\phi$ is satisfiable if there is a run $r$ of $T_\phi$ so that no Until subformula is stuck at $r$.
\end{theorem}

We have now shown that the temporal transition system $T_\phi$ is intimately related to the
satisfiability of $\phi$. The definition of $T_\phi$ is, however, rather nonconstructive.
In the next subsection we discuss how to construct $T_\phi$.

\subsection{System Construction}

First, we show how one can consider LTL formulas as propositional ones. This 
requires considering temporal subformulas as \textit{propositional atoms}. 
We now define the \textit{propositional atoms} of LTL formulas. 
 
\begin{definition}[Propositional Atoms]\label{def:propatoms}
For an LTL formula $\phi$, we define the set of \textit{propositional atoms} of $\phi$, 
i.e. $\u{\phi}$, as follows:
  \begin{enumerate}
    \item $\u{\phi}=\{\phi\}$ if $\phi$ is an atom, Next, Until or Release formula;
    \item $\u{\phi} = \u{\psi}$ if $\phi=(\neg\psi)$;
    \item $\u{\phi} = \u{\phi_1}\cup\u{\phi_2}$ if $\phi=(\phi_1\wedge\phi_2)$ or $\phi=(\phi_1\vee\phi_2)$.
  \end{enumerate}
\end{definition}

Consider, for example, the formula $\phi= (a\wedge (a U b)\wedge \neg (X (a\vee b)))$. 
Here we have $\u{\phi}$ is $\{a, (a U b), (X(a\vee b))\}$. Intuitively, the propositional atoms 
are obtained by treating all temporal subformulas of $\phi$ as atomic propositions. Thus,
an LTL formula $\phi$ can be viewed as a propositional formula over $\u{\phi}$. 

\begin{definition}
For an LTL formula $\phi$, let $\phi^p$ be $\phi$ considered as a propositional 
formula over $\u{\phi}$.
\end{definition}

We now introduce the \textit{neXt Normal Form} (XNF) of LTL formulas, which separates the 
``current'' and ``next-state'' parts of the formula, but costs only linear in the original formula size. 

\begin{definition}[neXt Normal Form]\label{def:xnf}
An LTL formula $\phi$ is in \emph{neXt Normal Form} (XNF) 
if there are no Unitl or Release subformulas of $\phi$ in $\u{\phi}$. 
\end{definition}

For example, $\phi= (a U b)$ is not in XNF,  while $(b\vee (a\wedge (X(a U b))))$ is in XNF.
Every LTL formula $\phi$ can be converted, with linear in the formula size, 
to an equivalent formula in XNF.
\begin{theorem}\label{thm:xnf}
For an LTL formula $\phi$, there is an equivalent formula $\f{\phi}$ that is in XNF.
Furthermore, the cost of the conversion is linear.
\end{theorem}
\begin{proof}
To construct $\f{\phi}$, We can apply the expansion rules 
$(\phi_1U\phi_2)\equiv(\phi_2\vee(\phi_1\wedge X(\phi_1U\phi_2)))$ and 
$(\phi_1R\phi_2)\equiv(\phi_2\wedge(\phi_1\vee X(\phi_1R\phi_2)))$. 
In detail, we can construct $\f{\phi}$ inductively:
  \begin{enumerate}
    \item $\f{\phi} = \phi$ if $\phi$ is $\tt$, $\ff$, a literal $l$ or a Next formula $X\psi$;
    \item $\f{\phi} = \f{\phi_1}\wedge\f{\phi_2}$ if $\phi=(\phi_1\wedge\phi_2)$;
    \item $\f{\phi} = \f{\phi_1}\vee\f{\phi_2}$ if $\phi=(\phi_1\vee\phi_2)$;
    \item $\f{\phi} = (\f{\phi_2})\vee(\f{\phi_1}\wedge X\phi)$ if $\phi=(\phi_1 U\phi_2)$;
    \item $\f{\phi} = \f{\phi_2}\wedge(\f{\phi_1}\vee X\phi)$ if $\phi=(\phi_1 R\phi_2)$.
  \end{enumerate} 
Since the construction is built on the two expansion rules that preserve the equivalence of formulas, 
it follows that $\phi$ is logically equivalent to $\f{\phi}$.
Note that the conversion map $\f{\phi}$ doubles the size of the converted formula $\phi$, but since
the conversion puts Until and Release subformulas in the scope of Next, and the conversion
stops when it comes to Next subformulas, the cost is at most linear.\qed
\end{proof}

We can now state propositional satisfiability of LTL formulas in terms of satisfiability
of propositional formulas. That is, by restricting LTL formulas to XNF, a satisfying assignment 
of $\phi^p$, which can be obtained by using a SAT solver, corresponds precisely to a propositional assignment of formula $\phi$.
\begin{theorem}\label{thm:assign}
For an LTL formula $\phi$ in XNF, if there is a satisfying assignment $A$ of $\phi^p$,
then there is a propositional assignment $A'$ of $\phi$ that satisfies $\phi$
such that $A'\cap\u{\phi}\subseteq A$. 
Conversely, if there is a propositional assignment $A'$ of $\phi$ that satisfies $\phi$, 
then there is a satisfying assignment $A$ of $\phi^p$ such that $A'\cap\u{\phi}\subseteq A$. 
\end{theorem}
\begin{proof}
($\Rightarrow$) 
Let $A$ be a satisfying assignment of $\phi^p$. Then let $A'$ 
be the set of all formulas $\psi\in cl(\phi)$ such that $A$ satisfies $(\f{\psi})^p$.
We clearly have that $A'\cap\u{\phi}\subseteq A$. 
According to Definition \ref{def:propsat} and because $\phi$ is in XNF, 
we have that $A'$ is a propositional assignment of $\phi$ that satisfies $\phi$.
  
($\Leftarrow$) 
Let $A'$ be a propositional assignment of $\phi$ that satisfies $\phi$.
Then let $A$ to be the assignment that assign true to $\psi\in cl(\phi)$
precisely when $\psi\in A'$. Again, we clearly have that, $A'\cap\u{\phi}\subseteq A$. 
According to Definition \ref{def:propsat} and because $\phi$ is in XNF, 
we have that $A$ is a satisfying assignment of $\phi^p$.\qed
\end{proof}

Theorem \ref{thm:assign} shows that by requiring the formula $\phi$ to be in XNF, we can construct 
the states of the transition system $T_{\phi}$ via computing satisfying assignments of $\phi^p$ over $\u{\phi}$. 
Let $t$ be a satisfying assignment of $\phi^p$ and $A_t$ be the related propositional assignment of $\phi$ 
generated from $t$ by Theorem \ref{thm:assign}, the construction is operated as follows: 
\begin{enumerate}
\item 
Let $S_0=\{A_t\ |\ t\models\phi^p\}$; and let $S:=S_0$,
\item 
Compute $S_i=\{A_t\ |\ t\models(\f{\bigwedge X(s_i)})^p\}$ for each $s_i\in S$, 
where $X(s_i)=\{\theta\ |\ (X\theta)\in s_i\}$; and update $S:=S\cup S_i$; 
\item 
Stop if $S$ does not change; else go back to step 2.
\end{enumerate}
The construction first generates initial states (step 1), and then all reachable states from initial ones (step 2);
it terminates once no new reachable state can be generated (step 3). 
So $S$ is the set of system states and its size is bounded by $2^{|cl(\phi)|}$.

Our goal here is to show that we can construct the transition system $T_\phi$ by means of SAT solving.
This requires us to refine Theorem~\ref{thm:reasoning}. A key issue in how a propositional assignment 
handles an Until formula is whether it satisfies it immediately or postpones it. We introduce
new propositions that indicate which is the case, and we refine the implementation of $\f$.
Given $\psi=(\psi_1 U \psi_2)$, we introduce a new proposition $v(\psi)$, and use the following conversion rule:
$\f{\psi}\equiv (v(\psi)\wedge\psi_2)\vee ((\neg v(\psi))\wedge \psi_1\wedge (X(\psi)))$. 
Thus, $v(\psi)$ is required to be true when the Until is satisfied immediately, and false when
the Until is postponed. Now we can state the refinement of Theorem~\ref{thm:reasoning}.

\begin{theorem}\label{thm:satreasoning}
For an LTL formula $\phi$, $\phi$ is satisfiable iff 
there is a \emph{finite} run $r=s_0,s_1,\ldots,s_n$ in $T_\phi$ such that
\begin{enumerate}
\item 
There are $0\leq m \leq n$ such that $s_m=s_n$;
\item 
Let $Q=\bigcup_{i= m}^n s_i$. If $\psi=(\psi_1U\psi_2)\in Q$, then $v(\psi)\in Q$.
\end{enumerate}
\end{theorem}  

\begin{proof}
Suppose first that items 1 and 2 hold. Then the infinite sequence 
$r'=s_0,\ldots,s_m,(s_{m+1},\ldots,s_n)^\omega$ is an infinite run of $T_\phi$.
It follows from Item 2 that no Until subformula is stuck at $r'$. By Theorem \ref{thm:reasoning},
we have that $r'\models\phi$.

Suppose now that $\phi$ is satisfiable. By Theorem \ref{thm:reasoning}, there is an infinite run
$r'$ of $T_\phi$ in which no Until subformula is stuck. Let $r'=s_0,s_1,\ldots$  be such a run. 
Each $s_i (i\geq 0)$ is a state of $T_\phi$, and the number of states is bounded by $2^{|cl(\phi)|}$. 
Thus, there must be $0\leq m< n$ such that $s_m=s_n$. Let $Q=\bigcup_{i= m}^n s_i$. Since no Until 
subformula can be stuck at $r$, if $\psi=\psi_1U\psi_2\in Q$, then it is must be that $v(\psi)\in Q$. \qed
\end{proof}

The significance of Theorem~\ref{thm:satreasoning} is that it reduces LTL
satisfiability checking to searching for a ``lasso'' in $T_\phi$ \cite{CVWY92}.
Item 1 says that we need to search for a prefix followed by a cycle, while
Item 2 provides a way to test that no Until subformla gets stuck in the infinite 
run in which the cycle $s_{m+1},\ldots,s_n$ is repeated infinitely often.

\subsection{Related Work}
We introduced our SAT-based reasoning approach above, and in this section we discuss the difference 
between our SAT-based approach and earlier works.

Earlier approach to transition-system construction for LTL formulas, 
based on tableau \cite{GPVW95} and normal form \cite{LZPVH13}, 
generates the system states explicitly or implicitly via a translation to \emph{disjunctive normal form} (DNF).
In  \cite{LZPVH13}, the conversion to DNF is explicit (though various heuristics are used to temper the exponential blow-up)
and the states generated correspond to the disjuncts.  In tableau-based tools, cf., \cite{GPVW95,DGV99}, the construction 
is based on iterative \emph{syntactic splitting} in which a state of the form $A\cup \{\theta_1\vee\theta_2\}$ 
is split to states: $A\cup \{\theta_1\}$ and $A\cup \{\theta_2\}$.

The approach proposed here is based on SAT solving, where the states correspond to satisfying assignments.
Satisfying assignments are generated via a search process that is guided by \emph{semantic splitting}. The
advantage of using SAT solving rather than syntactic approaches is the impressive progress in the development 
of heuristics that have evolved to yield highly efficient SAT solving: unit propagation, two-literal watching,
back jumping, clause learning, and more, see \cite{MZ09}. Furthermore, SAT solving continues to evolve in an impressive 
pace, driven by an annual competition%
\footnote{See \url{http://www.satcompetition.org/}}.
It should be remarked that an analogous debate, between syntactic and semantic approaches, took place
in the context of automated test-pattern generation for circuit designs, where, ultimately,
the semantic approach has been shown to be superior \cite{Larra92}.

Furthermore, relying on SAT solving as the underlying reasoning technology enables us to decouple
temporal reasoning from propositional reasoning. Temporal reasoning is accomplished via a search in the 
transition system, while the construction of the transition system, which requires proposition
reasoning using SAT solving. 

\section{LTL Satisfiability Checking}\label{sec:ltlsat}
Given an LTL formula $\phi$, the satisfiability problem is to ask whether there 
is an infinite trace $\xi$ such that $\xi\models\phi$.  In the previous section
we introduced a SAT-based LTL-reasoning framework and showed how it can be applied
to solve LTL reasoning problems. In this section we use this framework
to develop an efficient SAT-based algorithm for LTL satisfiability checking. 
We design a depth-first-search (DFS) algorithm that constructs the temporal 
transition system on the fly and searches for a trace per Theorem~\ref{thm:satreasoning}.
Furthermore, we propose several heuristics to reduce the search space.
Due to the limited space, we offer here a high-level description of the algorithms.
Details are provided in Appendix \ref{app:alg}.

\subsection{The Main Algorithm}
The main algorithm, {\sf LTL-CHECK}, creates the temporal transition system of the input formula 
on-the-fly, and searches for a lasso in a DFS  manner. Several prior works describe algorithms
for DFS lasso search , cf.~\cite{CVWY92,LZPVH13,SE05}. Here we focus on the steps that are 
specialized to our algorithm.

The key idea of LTL-CHECK is to create states and their successors using SAT techniques 
rather than traditional tableau or expansion techniques. Given the current formula $\phi$, 
we first compute its XNF version $\f{\phi}$, and then use a SAT solver to compute 
the satisfying assignments of $(\f{\phi})^p$. Let $P$ be a satisfying assignment for $(\f{\phi})^p$;
from the previous section we know that $X(P)=\{\theta\ |\ X\theta \in P\}$ yields a successor state in $T_\phi$. 
We implement this approach in the $getState$ function, which we improve later by introducing some heuristics. 
By enumerating all assignments of $(\f{\phi})^p$ we can obtain all successor states of $P$. Note, however 
that LTL-CHECK runs in the DFS manner, under which only a single state is needed at a time, so additional 
effort must be taken to maintain history information of the next-state generation for each state $P$.    

As soon as LTL-CHECK detects a lasso, it checks whether the lasso is accepting. Previous lasso-search
algorithms operate on the B\"uchi automaton generated from the input formula. In contrast, here
we focus directly on the satisfaction of Until subformulas per Theorem \ref{thm:satreasoning}. 
We use the example below to show the general idea. 

Consider the formula $\phi=G((F b)\wedge (F c))$.  
By Theorem \ref{thm:xnf}, $\f{\phi}=\f{Fb}\wedge \f{Fc}\wedge X\phi$, 
where $\f{Fb}=((b\wedge v(Fb))\vee (\neg v(Fb)\wedge X(Fb)))$ and 
$\f{Fc}=((c\wedge v(Fc))\vee (\neg v(Fc)\wedge X(Fc)))$. Suppose we get from the SAT solver
an assignment of $(\f{\phi})^p$ $P=\{v(Fb),\neg v(Fc), b, \neg c, \neg X(Fb), X(Fc), X\phi\}$. 
By Theorem \ref{thm:assign}, we create a satisfying assignment $A'$ that includes
all formulas in $cl(\phi)$ that are satisfied by $P$, and we get the state
$s_0= P \cup \{\phi,Fb,Fc,(Fb)\wedge (Fc)\}$. To obtain the next state, we start
with $X(s_0)=\{Fc,\phi\}$, compute $\f{Fc\wedge\phi}$ and repeat the process.
After several steps LTL-CHECK may find a path $s_0\tran{}s_1\tran{}s_0$, 
where $s_1=\{\phi, Fb, Fc, (Fb)\wedge (Fc), \neg v(Fb), v(Fc), \neg b, c, X(F b), \neg X(Fc), X\phi\}$.
Now $s_0$ and $s_1$ form a lasso. Let $Q=s_0\cup s_1$. Both $Fb$ and $Fc$ are in $Q$,
and also $v(Fb)$ and $v(Fc)$ are in $Q$. By Theorem \ref{thm:satreasoning}, $\phi$ is satisfiable.

\subsection{Heuristics for State Elimination}
While LTL-CHECK uses an efficient SAT solver to compute states of the system in the $getState$ function, 
this approach is effective in creating states and their successors, but cannot be used to guide
the overall search. To find a satisfying lasso faster, we add heuristics that drive the search
towards satisfaction.  The key to these heuristics is smartly choosing the next state given by SAT solvers. 
This can be achieved by adding more constraints to the SAT solver. Experiments show these
heuristics are critical to the performance of our LTL-satisfiability tool.

The construction of state in the transition system always starts with formulas.
At the beginning, we have the input formula $\phi_0$ and we take the following steps:
(1) Compute $\f{\phi_0}$; (2) Call a SAT solver to get an assignment $P_0$ of $(\f{\phi_0})^p$; 
and (3) Derive a state $P'_0$ from $P_0$. Then, to get a successor state, we start with the formula
$\phi_1=\bigwedge X(P'_0)$, and repeat steps (1-3). Thus, every state $s$ is obtained 
from some formula $\phi_s$, which we call the \emph{representative formula}. 
Note that with the possible exception of $\phi_0$, all representative formulas are conjunctions.
Let $\phi_s=\bigwedge_{1\leq i\leq n}\theta_i$ be the representative formula of a state $s$;
we say that $\theta_i (1\leq i\leq n)$ is an \emph{obligation} of $\phi$ if $\theta_i$ is an Until formula. 
Thus, we associate with the state $s$ a set of obligations, which are the Until conjunctive elements of $\phi_s$.
(The initial state may have obligations if it is a conjunction.)
The approach we now describe is to satisfy {obligations as early as possible during the search, 
so that a satisfying lasso is obtained earlier.  We now refine the $getState$ function, 
and introduce three heuristics via examples.

\begin{figure}[h!]
\centering
\includegraphics[scale=1]{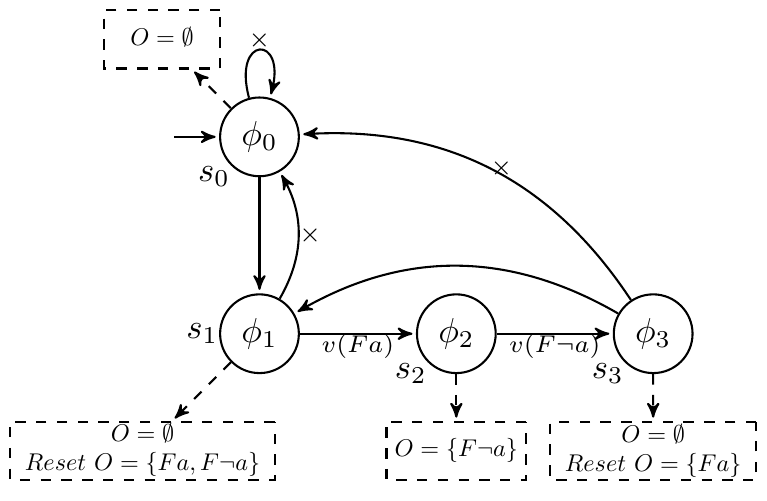}
\caption{A satisfiable formula. 
In the figure $\phi_0=G((Fa)\wedge (F\neg a))$, 
$\phi_1=((Fa) \wedge (F\neg a)\wedge \phi_0)$, 
$\phi_2= ((F\neg a)\wedge\phi_0)$ and 
$\phi_3=((Fa)\wedge\phi_0)$. 
These representative formulas correspond to states $s_0,s_1,s_2,s_3$, respectively.}
\label{fig:sat_example}
\end{figure}

The $getState$ function keeps a global \emph{obligation set}, collecting all obligations 
so far not satisfied in the search. The obligation set is initialized with the obligations 
of the initial formula $\phi_0$. When an obligation $o$ is satisfied (i.e., when $v(o)$ is true),
$o$ is removed from the obligation set.  Once the obligation set becomes empty in the search,
it is reset to contain obligations of current representative formula $\phi_i$. 
In Fig. \ref{fig:sat_example}, we denote the obligation set by $O$.
$O$ is initialized to $\emptyset$, as there is no obligation in $\phi_0$. 
$O$ is then reset in the states $s_1$ and $s_3$, when it becomes empty.

The $getState$ function runs in the \textbf{\ue} mode by default, in which it obtains 
the next state guided by the obligations of current state. For satisfiable formulas, this 
leads to faster lasso detection.  Consider formula $\phi=G((F a) \wedge (F\neg a))$. 
Parts of the temporal transition system $T_\phi$ are shown in Fig. \ref{fig:sat_example}. 
In the figure, $O$ is reset to $\{(Fa), (F\neg a)\}$ in state $s_1$, as these are the obligations 
of $\phi_1$. To drive the search towards early satisfaction of obligations, we obtain
a successor of $s_1$, by applying the SAT solver to the formula $(\f{\phi_1}\wedge (v(Fa)\vee v(F\neg a)))^p$, 
to check whether $Fa$ or $F\neg a$ can be satisfied immediately.  
If the returned assignment satisfies $v(Fa)$, then we get the success state $s_2$ with the
representative formulas $\phi_2$, and $(Fa)$ is removed from $O$. Then the next state is $s_3$
with the representative formula $\phi_3$, which removes the obligation $(F\neg a)$. 
since $O$ becomes empty, it is reset to the obligations $\{Fa\}$ of $\phi_3$.
Note that in Fig. \ref{fig:sat_example}, there should be transitions from $s_2$ to $s_1$ 
and from $s_3$ to $s_2$,  but they are never traversed under the \ue mode. 

The $getState$ function runs in the \textbf{\sp} mode when the obligation set 
becomes empty. In this mode, we want to check whether the next state can be
a state that have been visited before and after that visit the obligation set 
has become empty.  In this case, the generated lasso is accepting, by Theorem~\ref{thm:satreasoning}.
In Fig. \ref{fig:sat_example}, the obligation set $O$ becomes empty in state $s_3$.
Previously, it has become empty in $s_1$. Normally, we find a success state for $s_3$ 
by applying the SAT solver to $(\f{\phi_3})^p$.  To find out if either $s_0$ or $s_1$ 
can be a successor of $s_3$, we apply the SAT solver to the formula 
$(\f{\phi_3}\wedge (X(\phi_0)\vee X(\phi_1)))^p$. Since this formula is satisfiable and
indicates a transition from $s_3$ to $s_1$ ($X\phi_1$ can be assigned true in the assignment), 
we have found that $trace(s_0),(trace(s_1),trace(s_2),trace(s_3))^\omega$ satisfies $\phi$.
In the figure, the transitions labeled \textbf{x} represent failed attempts to generate 
the lasso when $O$ becomes empty. Although failed attempts have a computational cost, 
trying to close cycles aggressively does pay off.


The $getState$ function runs in the \textbf{\up} mode if all formulas in the obligation 
set are postponed in the \ue mode. The goal of this mode is to eliminate ``conflicts'' 
that block immediate satisfaction of obligations.
To achieve this, we use a \textit{conflict-guided} strategy. 
Consider, for example, the formula $\phi_0=a\wedge (Xb)\wedge F((\neg a)\wedge (\neg b))$.
Here the formula $\psi=F((\neg a)\wedge (\neg b))$ is an obligation.  We check whether $\psi$ 
can be satisfied immediately, but it fails.  The reason for this failure is the conjunct $a$ in $\phi$, 
which conflicts with the obligation $\psi$. We identify this conflict using 
a \textit{minimal unsat core} algorithm \cite{MSL11}. To eliminate this conflict, we add the conjunct
$\neg Xa$ to $\phi$, hoping to be able to satisfy the obligation immediately in the next state.
When we apply the SAT solver to $(\f{\phi}\wedge (\neg Xa))^p$, we obtain a successor state with the
representative formula $\phi_1=(b\wedge\psi)$, again with $\psi$ as an obligation.
When we try to satisfy $\psi$ immediately, we fail again, since $\psi$ conflicts with $b$.
To block both conflicts, we add $\neg Xb$ as an additional constraint, and apply the SAT solver to
$(\f{\phi}\wedge (\neg Xa)\wedge (\neg Xb))^p$. This yields a successor state with the representative formula
$\phi_2=\psi$. Now we are able to satisfy $\psi$ immediately, and we are able to satisfy $\phi$
with the finite path $\phi\tran{}\phi_1\tran{}\phi_2$.

As another example, consider the formula $\phi=(G(Fa) \wedge Gb\wedge F(\neg b))$.
Since $F(\neg b)$ is an obligation, we try to satisfy it immediately, but fail. The
reason for the failure is that immediate satisfaction of $F(\neg b)$ conflicts with the
conjunct $Gb$. In order to try to block this conflict, we add to $\phi$ the conjunct
$\neg XGb$, and apply the SAT solver to $(\f{\phi}\wedge\neg XGb)^p$.  This also fails.
Furthermore, by constructing a minimal unsat core, we discover that $(\f{Gb}\wedge \neg X(Gb))^p$ 
is unsatisfiable.  This indicates that $Gb$ is an ``invariant''; that is, if $Gb$ is true in a state
then it is also true in its successor. This means that the obligation $F(\neg b)$ can never be
satisfied, since the conflict can never be removed. Thus, we can conclude that $\phi$ is
unsatisfiable without constructing more than one state.

In general, identifying conflicts using minimal unsat cores enables both to find satisfying traces faster,
or conclude faster that such traces cannot be found.

\section{Experiments on LTL Satisfiability Checking}\label{sec:experiment}

In this section we discuss the experimental evaluation for LTL satisfiability checking. 
We first describe the methodology used in experiments and then show the results.

\subsection{Experimental Methodologies}
The platform used in the experiments is an IBM iDataPlex consisting of
2304 processor cores in 192 Westmere nodes (12 processor cores per node) at 2.83 GHz 
with 48 GB of RAM per node (4 GB per core), running the 64-bit Redhat 7 operating system.
In our experiments, each tool runs on a single core in a single node. 
We use the Linux command ``time'' to evaluate the time cost (in seconds) of each experiment. 
Timeout was set to 60 seconds, and the out-of-time cases are set to cost 60s.  

We implemented the satisfiability-checking algorithms introduced in this paper, 
and named the tool \tool%
\footnote{It can be downloaded at \url{www.lab205.org/aalta}.}. 
We compare \tool with Aalta\_v1.2, 
which is the latest explicit LTL-satisfiability solver (though it does use
some SAT solving for acceleration) \cite{LPZHVCoRR14}. (The SAT engine used in both 
Aalta\_v1.2 and \tool is Minisat \cite{ES03}.) In the literature, Aalta\_v1.2 is shown to 
outperform other existing explicit LTL solvers, 
so we omit the comparison with these solvers in this paper. 
Two resolution-based LTL satisfiability solvers, TRP++ \cite{HK03} and ls4 \cite{SW12},
are also included in our comparison. (Note ls4 utilizes SAT solving as well.)

As shown in \cite{RV10}, LTL satisfiability checking can be reduced to model checking.
While BDD-based model checker were shown to be competitive for LTL satisfiability solving
in \cite{RV10}, they were shown later not to be competitive with specialized tools, such as Aalta\_v1.2 \cite{LZPVH13}.
We do, however, include in our comparison the model checker NuXmv \cite{CCDGMMMRT14}, 
which integrates the latest SAT-based model checking techniques. It uses Minisat as the SAT engine as well.  Although standard bounded model checking (BMC) 
is not complete for the LTL satisfiability checking, there are techniques to make it complete, 
for example, \emph{incremental bounded model checking} (BMC-INC) \cite{KHL05}, which is implemented in NuXmv.
In addition, NuXmv implements also new SAT-based techniques, IC3 \cite{Bra11}, which 
can handle liveness properties with the K-liveness technique\cite{CS12}. We included
IC3 with K-liveness in our comparison.

To compare with the K-liveness checking algorithm, we ran NuXmv using the command ``check\_ltlspec\_klive -d''. 
For the BMC-INC comparison, we run NuXmv with the command ``check\_ltlspec\_sbmc\_inc -c''. 
\tool , Aalta\_v1.2 and ls4 tools were run using their default parameters, while TRP++ runs with ``-sBFS -FSR''. 
Since the input of TRP++ and ls4 must be in SNF (Separated Normal Form \cite{Fish97}), 
an SNF generator is required for running these tools. We use the generator \textit{TST-translate} which belongs to 
ls4 tool suit. 

In the experiments we consider the benchmark suite from \cite{SD11}, referred to as \textit{schuppan-collected}. 
This suite collects formulas from several prior works, including \cite{RV10},
and has a total of 7446 formulas (3723 representative formulas and their negations). 
(Testing also the negation of each formula is in essence a check for validity.)
In our experiments, we did not find any inconsistency among the solvers that did not time out.

\subsection{Results}

\begin{table*}[htb]
\centering
    \caption{Experimental results on the Schuppan-collected benchmark. 
     Each cell lists a tuple $\langle t, n\rangle$ where $t$ is the total checking time (in seconds), 
     and $n$ is the total number of unsolved formulas.}\label{tab:result2}
    \scalebox{0.7}
    {
    \begin{tabular}{|l|l|r||l|r||l|r||l|r||l|r||l|r||l|r||}
    \hline
    Formula type  & \multicolumn{2}{c|}{ls4} & \multicolumn{2}{c|}{TRP++} & \multicolumn{2}{c|}{\pbox{2cm}{NuXmv-\\BMCINC}} &  \multicolumn{2}{c|}{Aalta\_v1.2} & \multicolumn{2}{c|}{\pbox{2cm}{NuXmv-\\IC3-Klive}} & \multicolumn{2}{c|}{\pbox{3cm}{\tool \\without heuristics}} & \multicolumn{2}{c|}{\pbox{2.5cm}{\tool \\with heuristics}}\\
    \hline
    /acacia/example  &  155 & 0 &  192 & 0 &  1 & 0 & 1 & 0 &  8 & 0 & 1 & 0 &  1 & 0\\
    \hline
    /acacia/demo-v3  &  68  & 0 &  2834  & 38 & 3  & 0 & 660 & 0 & 30  & 0 & 630 & 0 & 3 & 0\\
    \hline
    /acacia/demo-v22  &  60  & 0 &  67  & 0 &  1 & 0 &  2 & 0 & 4 & 0 &  2 & 0 &  1 & 0\\
    \hline
    /alaska/lift  &  2381 & 27 &  15602 & 254 &  1919 & 26 & 4084  & 63 &  867 & 5 & 4610  & 70 &  1431 & 18\\
    \hline  
    /alaska/szymanski  &  27 & 0 &  283 & 4 &  1 & 0 & 1  & 0 &  2 & 0 & 1  & 0 &  1 & 0\\
    \hline
    /anzu/amba  &  5820 & 92 &  6120 & 102 &  536 & 7 &  2686 & 40 & 1062 & 8 &  3876 & 60 &  928 & 4\\
    \hline 
    /anzu/genbuf  &  2200 & 30 &  7200 & 120 &  782 & 11 &  3343 & 54 & 1350 & 13 &  5243 & 94 &  827 & 4\\
    \hline
    /rozier/counter  &  3934 & 62 &  4491 & 44 &  3865 & 64 &  3928 & 60 & 3988 & 65 &  3328 & 55 &  2649 & 40\\
    \hline  
    /rozier/formulas  &  167 & 0 &  37533 & 523 &  1258 & 19 &  1372 & 20 & 664 & 0 &  1672 & 25 &  363 & 0\\
    \hline  
    /rozier/pattern  &  2216 & 38 &  15450 & 237 &  1505 & 8 & 8  & 0 &  3252 & 17 & 8  & 0 &  9 & 0\\
    \hline  
    /schuppan/O1formula  &  2193 & 34 &  2178 & 35 &  14 & 0 & 2 & 0 & 95  & 0 & 2 & 0 & 2 & 0\\
    \hline  
    /schuppan/O2formula  &  2284 & 35 &  2566 & 41 &  1781 & 28 &  2 & 0 & 742 & 7 &  2 & 0 &  2 & 0\\
    \hline  
    /schuppan/phltl  &  1771 & 27 &  1793 & 29 &  1058 & 15 &  1233 & 21 & 753 & 11 &  1333 & 21 &  767 & 13\\
    \hline  
    /trp/N5x  &  144 & 0 &  46 & 0 &  567 & 9 &  309 & 0 & 187 & 0 &  219 & 0 &  15 & 0\\
    \hline 
    /trp/N5y  &  448 & 10 &  95 & 1 &  2768 & 46 & 116 & 0 & 102 & 0 & 316 & 0 &  16 & 0\\
    \hline  
    /trp/N12x  &  3345 & 52 &  45739 & 735 &  3570 & 58 &  768 & 48 & 705 & 0 &  768 & 0 &  175 & 0\\
    \hline  
    /trp/N12y  & 3811 & 56 &  19142 & 265 &  4049 & 67 &  7413 & 110 & 979 & 0 &  7413 & 100 &  154 & 0\\
    \hline  
    /forobots  &  990 & 0 &  1303 & 0 &  1085 & 18 &  2280 & 32 & 37 & 0 &  2130 & 30 &  524 & 0\\
    \hline 
    Total & 32014 & 463 & 163142 & 2428 & 24769 & 376 & 31208 & 450 & 14261 & 126 & 31554 & 455 & 7868 & 79\\
    \hline
\end{tabular}
}
\end{table*}

The experimental results are shown in Table \ref{tab:result2}. 
In the table, the first column lists the different benchmarks in the suite, 
and the second to eighth columns display the results from different solvers. 
Each result in a cell of the table is a tuple $\langle t, n\rangle$, 
where $t$ is the total checking time for the corresponding benchmark, 
and $n$ is the number of unsolved formulas due to timeout in the benchmark. 
Specially the number ``0'' in the table means all formulas in the given benchmark are solved. 
Finally, the last row of the table lists the total checking time and number 
of unsolved formulas for each solver. 

The results show that while the tableau-based tool Aalta\_v1.2, outperforms ls4 and TRP++,
it is outperformed by NuXmv-BMCINC and NuXmv-IC3-Klive, both of which are outperformed
by \tool, which is faster by about 6,000 seconds and solves 47 more instances than
NuXmv-IC3-Klive. 

Our framework is explicit and closest to that is underlaid behind Aalta\_v1.2. From the results, 
\tool with heuristic outperforms Aalta\_v1.2 dramatically, faster by more than 23,000 seconds 
and solving 371 more instances. One reason is, when Aalta\_v1.2 fails it is often due to
timeout during the heavy-duty normal-form generation, which \tool simply avoids
(generating XNF is rather lightweight). 

Generating the states in a lightweight way, however, is not efficient enough. By running \tool without 
heuristics, it cannot perform better than Aalta\_v1.2, see the data in column 5 and 7 of Table \ref{tab:result2}. 
It can even be worse in some 
benchmarks such as ``/anzu/amba'' and ``anzu/genbuf''. 
We can explain the reason via an example. 
Assume the formula is $\phi_1\vee\phi_2$, the traditional tableau method splits the formula and at most creates two nodes. 
Under our pure SAT-reasoning framework, however,it may create three nodes which contain $\phi_1\wedge\neg\phi_1$ or 
$\neg\phi_1\wedge\phi_2$, or $\phi_1\wedge\phi_2$. 
This indicates that the state space generated by SAT solvers may in general be larger than that generated 
by tableau expansion. 

To overcome this challenge, we propose some heuristics by adding specific constraints to SAT solvers, which at the mean 
time succeeds to reduce the searching space of the overall system.   
The results shown in column 8 of Table \ref{tab:result2} demonstrate the effectiveness of heuristics presented in the paper. 
For example, the ``/trp/N12/'' and ``/forobots/'' benchmarks are mostly unsatisfiable 
formulas, which Aalta\_v1.2 and \tool with heuristic do not handle well. Yet the \textit{unsat-core extraction} heuristic, 
which is described in the \up mode of $getState$ function, enables \tool with heuristic to solve 
all these formulas. For satisfiable formulas, the results from ``/anzu/amba/'' and ``/anzu/\\genbuf'' formulas, 
which are satisfiable, show the efficiency of the \ue and \sp heuristics in the $getState$ function,
which are necessary to solve the formulas. 

In summary, 
\tool with heuristic performed best on satisfiable formulas, solving 6750 instances, followed in order by NuXmv-BMCIMC (6714), 
NuXmv-IC3-Klive (6700), Aalta\_v1.2 (6689), ls4 (6648), and TRP++ (4711). For unsatisfiable formulas, 
NuXmv-IC3-Klive performs best, solving 620 instances, followed in order by \tool with heuristic (617), 
NuXmv-BMCINC (356), ls4 (335), Aalta\_v1.2 (309), and TRP++ (307). Detailed statistics are in Appendix \ref{app:detail_results}. 

Note that NuXmv-IC3-Klive is able to solve more cases than \tool with heuristic in some benchmarks, 
such as ``/lift'' and ``/schuppan/phltl'' in which unsatisfiable formulas are not handled 
well enough by \tool. Currently, \tool requires  large number of SAT calls to
identify an unsatisfiable core. In future work we plan to use a specialized MUS (minimal unsatisfable core)
solver to address this challenge.

\section{Concluding Remarks}\label{sec:discussion}

We described in this paper a SAT-based framework for explicit LTL reasoning.
We showed one of its applicaitons to LTL-satisfiability checking, by proposing basic 
algorithms and efficient heuristics.  As proof of concept, we implemented an 
LTL satisfiability solver, whose performance dominates all similar tools. 
In Appendix \ref{app:smt} we demonstrate that our approach can be extended from propositional 
LTL to assertional LTL, yielding exponential improvement in performance.

Extending the explicit SAT-based approach to other applications of LTL reasoning,
is a promising research direction. For example, the standard approach in LTL model 
checking \cite{VW86b} relies on the translation of LTL formulas to B\"uchi automata.
The transition systems $T_\phi$ that is used for LTL satisfiability checking
can also be used in the translation from LTL to B\"uchi automata.
Current best-of-breed translators, e.g., \cite{DP04,GPVW95,DGV99,SB00}
are tableau-based, and the SAT approach may yield significant performance improvement.

Of course, the ultimate temporal-reasoning task is model checking. Explicit model
checkers such as SPIN \cite{Hol03} start with a translation of LTL to B\"uchi 
automata, which are then used by the model-checking algorithm. An alternative
approach is to construct the automaton on-the-fly using SAT techniques,
using the framework developed here.  Current symbolic model-checking tools, such as NuXmv, 
do rely heavily on SAT solvers to implement algorithms such as BMC \cite{KHL05} or IC3 \cite{Bra11}.
The success of the SAT-based explicit LTL-reasoning approach for LTL satisfiability
checking suggests that this approach may also be successful in SAT-based model checking.
This remains a highly intriguing research possibility.

\bibliographystyle{plain}
\bibliography{ok,cav}

\appendix

\section{Proof of Theorem \ref{thm:propassign}}
\begin{proof}
If $A$ propositionally satisfies $\phi$ and $\xi\models \bigwedge A$, then $\xi\models\phi$, as $\phi\in A$.

For the other direction, assume that $\xi\models\phi$. Let $A=\{\theta\in cl(\phi): \xi\models\theta\}$.
Clearly, $\phi\in A$. It remains to prove that $A$ is a propositional assignment,
which we show by structural induction. 
\begin{itemize}
\item 
For $\ell\in L$ either $\xi\models\ell$ or $\xi\not\models\ell$, so
either $\ell\in A$ or $(\neg\ell)\in A$.
\item 
If $\xi\models (\theta_1\wedge\theta_2)$, then $\xi\models\theta_1$ and $\xi\models\theta_2$,
so both $\theta_1\in A$ and $\theta_2\in A$.
\item 
If $\xi\models (\theta_1\vee\theta_2)$, then $\xi\models\theta_1$ or $\xi\models\theta_2$,
so either $\theta_1\in A$ or $\theta_2\in A$.
\item
If $\xi\models(\theta_1 U\theta_2)$, then either $\xi\models\theta_2$, in which case, $\theta_2\in A$,
or $\xi\models \theta_1$ and $\xi\models (X(\theta_1 U \theta_2)$, in which case $\theta_1\in A$
and $(X(\theta_1 U \theta_2))\in A$.
\item
If $\xi\models(\theta_1 R\theta_2)$, then $\xi\models\theta_2$, in which case $\theta_2\in A$, and
either $\xi\models\theta_1$ and $\xi\models (X(\theta_1 R \theta_2))$, in which case
$\theta_1\in A$ and $(X(\theta_1 R \theta_2))\in A$.
\end{itemize}\qed
\end{proof}

\section{Proof of Theorem \ref{thm:reasoning}}

\begin{proof}
for the first claim, let $r$ be $s_0,s_1,\ldots$ and $r_i=s_i,s_{i+1},\ldots(i\geq 0)$. 
Assume that no Until subformua is stuck at $r$. We prove by induction that $trace(r_i)\models\psi$ for $\psi\in s_i$.
It follows that $trace(r)\models \phi$.
\begin{itemize}
\item
Trivially, for a literal $\ell\in s_i$ we have that $trace(r_i)\models\ell$.
\item
If $(\theta_1\wedge\theta_2)\in s_i$ , then $\theta_1\in s_i$ and $\theta_2\in s_i$.
By induction, $trace(r_i)\models \theta_1$ and $trace(r_i)\models \theta_2$,
so $trace(r_i)\models (\theta_1\wedge\theta_2)$.
The argument for $(\theta_1\vee\theta_2)\in s_i$ is analogous.
\item
If $(X\theta)\in s_i$, then $\theta\in s_{i+1}$. By induction, $trace(r_{i+1})\models \theta$,
so $trace(r_i) \models (X\theta)$
\item 
If $(\theta_1 U \theta_2)\in s_i$, then $\theta_2\in s_i$ or both $\theta_1\in s_i$
and $(X(\theta_1 U \theta_2))\in s_i$, which implies that $(\theta_1 U \theta_2)\in s_{i+1}$. 
Since  $(\theta_1 U \theta_2)$ is not stuck at $r$, there is some $k\geq i$ such that
$\theta_2\in s_k$, and $\theta_i\in s_j$ for $i\leq j \leq k$. Using the induction hypothesis
and the semantics of Until, it follows that $trace(r_i)\models  (\theta_1 U \theta_2)$.
\item 
If $(\theta_1 R \theta_2)\in s_i$, then $\theta_2\in s_i$ and either $\theta_1\in s_i$
or $(X(\theta_1 R \theta_2))\in s_i$, which implies $(\theta_1 R \theta_2)\in s_{i+1}$.
It is possible here for $(\theta_1 R \theta_2)$ to be postponed forever.
So for all $k\geq i$, we have that either $\theta_2\in s_j$ or
there exists $i\leq j \leq k$ such that $\theta_i\in s_j$. Using the induction hypothesis
and the semantics of Release, it follows that $trace(r_i)\models  (\theta_1 R \theta_2)$.
\end{itemize}

It follows that if there is a run $r$ of $T_\phi$ such that no Until subformul is stuck at $r$
then $\phi$ is satisfiable.

In the other direction, assume that $\phi$ is satisfiable and there is an infinite trace 
$\xi\in L^\omega$ such that $\xi\models\phi$. Let $\xi=P_0,P_1,\ldots$, and let $\xi_i=P_i,P_{i+1},\ldots$. 
As in the proof of Theorem \ref{thm:propassign}, define $A_i=\{\theta\in cl(\phi): \xi_i\models\theta\}$.
As in the proof of Theorem \ref{thm:propassign}, each $A_i$ is a propositional assignment for $\phi$,
and, consequently a state of $T_\phi$. Furthermore, the semantics of Next implies that
we have $T(A_i,A_{i+1})$ for $i\geq 0$. Furthermore, the semantics of Until ensures that no
Until is stuck in the run $A_0,A_1,\ldots$. \qed
\end{proof}

\section{Implementation of LTL-Satisfiability Checking Algorithms}\label{app:alg}

\subsection{Main checking algorithm}
The main algorithm checks the satisfiability of the input formula on the fly. 
It implements a \emph{depth-first search} to identify the lasso described in Theorem~\ref{thm:satreasoning}.
Algorithm \ref{alg:main} shows the details of the main algorithm which is named LTL-CHECK.

\begin{algorithm}
\caption{LTL Main Checking Algorithm: LTL-CHECK}\label{alg:main}
  \begin{algorithmic}[1]      
       \REQUIRE An LTL formula $\phi$.
       \ENSURE  SAT or UNSAT.
       \IF {$\phi=\tt$ (or $\phi=\ff$)}
         \RETURN SAT (or UNSAT);
       \ENDIF
       \STATE Let $\phi=\f{\phi}$: make $\phi$ ready for SAT solver;
       \STATE CALL $getState(\phi)$: get one system state $P$ from $\phi$;
       \WHILE {$P$ is existed}
         \STATE Let $\psi = \bigwedge X(P)$ be the next state of $\phi$;
         \IF {$\psi$ is in $explored$}
           \STATE CALL $getState(\phi)$ again: get another $P$;
           \STATE Continue;
         \ENDIF
         \IF {$\psi$ is visited}
           \IF {$model(\psi)$ is true}
             \RETURN SAT;
           \ENDIF
         \ELSE
           \STATE Push $\psi$ to $visited_S$, and push $P$ to $visited_P$;
           \IF {LTL-CHECK ($\psi$) is SAT}
             \RETURN SAT;
           \ENDIF
           \STATE Pop $\psi$ from $visited_S$, and pop $P$ from $visited_P$;
         \ENDIF
         \STATE CALL $getState(\phi)$ again: get another $P$;
       \ENDWHILE
       \STATE Push $\phi$ to $explored$;
       \RETURN UNSAT;
     \end{algorithmic}
\end{algorithm}

In Algorithm \ref{alg:main}, the function $\f{\phi}$ (in Line 4) is implemented according to Theorem \ref{thm:xnf}. 
It returns the next normal form of $\phi$. The function $getState$ takes an input LTL formula $\phi$ and outputs 
another system state $\phi'$ from $\phi$. We have that $T(CF(\phi), X(CF(\phi')))$, which means that $X(CF(\phi'))$ 
is one of next states of $\phi$. As mentioned previously, these can be obtained from the assignments of $(\f{\phi})^p$. 
Another main task of $getState$ is to return a different state never returned before in every invocation. 
More details are shown in Algorithm \ref{alg:getnextpair}. 

The main algorithm maintains three global lists: $visited_S$, $visited_P$ and $explored$, which record visited state, visited assignments and explored states respectively. So, $visited_S[i+1]$ is a next state of $visited_S[i]$ ($i\geq 0$), and $visited_P[i+1]$ is an assignment of $\bigwedge X(visited_P(i))$. Note explored states are those all of the successors are visited but no satisfying model is found. So explored states are unsatisfiable formulas. The function $model$ function (in Line 13) is to check whether the cycle found (containing $\psi$) is accepting. It is evaluated according to Theorem \ref{thm:satreasoning}.

 \begin{algorithm}
 \caption{Implementation of $getState$}\label{alg:getnextpair}
    \begin{algorithmic}[1]      
    \REQUIRE an LTL formula $\phi$;
    \ENSURE  a new state of $\phi$;
       \STATE Let $\alpha$ be \\$\phi^p\wedge(\bigwedge_{\psi\in explored}\neg (X\psi)^p)^1\wedge(\bigwedge_{\psi\in history}\neg \psi^p)^2$;
       \IF {$\alpha$ is satisfiable}
         \STATE Let $P$ be an assignment of $\alpha$;
         \STATE $history = history \cup \{\bigwedge P\}$;
         \RETURN $P$;
       \ELSE
         \RETURN $null$
       \ENDIF
       
     \end{algorithmic}
 \end{algorithm}
 
At the very beginning CHECK checks whether the formula is $\tt$ or $\ff$ (Line 1-3), in which cases the satisfiability can be determined immediately. Then it computes the next normal form of input formula $\phi$ (Line 4), acquiring a state $P$ from $(\f{\phi})^p$ (Line 5). If $P$ is not existed (Line 6), i.e. $\phi$ is checked unsatisfiable after exploring all its next states, then it is pushed to $explored$ (Line 25) and LTL-CHECK returns UNSAT (Line 26). Otherwise, LTL-CHECK makes sure that the chosen new next state $\psi$ of $\phi$ is not explored (Line 8-11). Later it checks whether $\psi$ has been visited before (Line 12). If so a cycle has been found and the $model$ function is invoked to check whether a satisfying model is found as well (Line 13-15). If this fails then another $P$ is required for further checking (Line 23). If $\psi$ is not visited yet, it is pushed into $visited$ and CHECK is invoked recursively by taking $\psi$ as the new input (Line 17,18). If $\psi$ is checked to be SAT so does $\phi$, else $\psi$ is popped from $visited$ (Line 21) and CHECK selects another $P$ for further checking (Line 23). One may find the algorithm can terminate as soon as all states from $(\f{\phi})^p$ are constructed. 

The task of $getState$ is not only to return a system state for the input formula, but also to guarantee every invoke by taking $\phi$ as input it does not return the state already created. To achieve this, we introduce another set $history$ to store all states already created so far for each current formula $\phi$. Then the assignment of $\alpha$ in Algorithm \ref{alg:getnextpair} can make the assignment distinguished with those ever created before. Note in Line 1, the expression labeled 1 erases those states already explored, which are shown to be unsatisfiable. And the expression labeled 2 guarantees those assignments that already appeared before cannot be chosen again. By adding these two constraints it avoids SAT solvers to create duplicated assignments. The notation $null$ in Line 7 represents the state required is not existed.

However, simply avoidance to generate duplicated states is not efficient enough for checking on a state space that is exponential larger than the original size of input formula. In the following section we present some heuristics to guide SAT solvers to return the assignment we prefer as soon as possible.

\subsection{Guided State Generation}  

Recall our basic reasoning theorem (Theorem \ref{thm:satreasoning}), the principle we judge whether a cycle can form a satisfying model is to check the satisfaction of Until formulas in $CF(\psi)$\footnote{$CF(\psi)$ denotes the set of conjuncts of $\psi$ by taking it as an And formula.}, where $\psi$ is a state in the cycle. So an intuitive idea to speed up the checking process is to locate such a satisfying cycle as soon as possible. As the satisfying cycle keeps satisfying the satisfaction of Until formulas, we follow this way and always try to ask SAT solvers to return assignments that can satisfy some Until formulas in $CF(\psi)$ where $\psi$ is the current state.

\begin{algorithm}
\caption{Implementation of the \ue mode}\label{alg:ue}  
\begin{algorithmic}[1]  
\REQUIRE an LTL formula $\phi$ and a global set $U$; 
\ENSURE  a new state of $\phi$;
       \IF {U is $\emptyset$}
         \STATE Turn into the \sp mode.
         \STATE Reset U to be $U(\phi)$, where $\phi$ is the current state and $U(\phi)\subseteq CF(\phi)$ is the set of Until formulas.
       \ENDIF
       \STATE Let P be an assignment of \\
       $(\f{\phi})^p\wedge\bigvee V(U)\wedge\bigwedge_{\psi\in explored}\neg (X(\psi))^p$, where 
       $V(U)=\{v(u)\ |\ u\in U\}$;
       \IF {P is empty}
         \STATE Turn into the \up mode.
       \ELSE
         \STATE Update $U=U\backslash S$, where $S=\{u\ |\ v(u)\in P\textit{ and }u\in U\}$;
         \RETURN P;
       \ENDIF
  \end{algorithmic}
\end{algorithm}

We now redesign the $getState$ function in three modes, which focus on different tasks. The \ue mode tries to fulfill the satisfaction of Until formulas in a global set, the \sp mode is to pursue a satisfying cycle, and the \up mode is to pursue an unsatisfiable core if all Until formulas remained in the set are postponed. The $getState$ function runs in the \ue mode by default. The implementation of the \ue mode is shown in Algorithm \ref{alg:ue}.

In the \ue mode a global get $U$ is used to keep the Until formulas postponed so far. It is initialized as $U(\phi)$, which is the set of Until formulas in $CF(\phi)$ ($\phi$ is the input formula). The task of the \ue mode is to check whether some elements in $U$ can be satisfied (Line 5). If $U$ becomes empty, the \sp mode is invoked to seek a satisfiable cycle (Line 1,2). If the \sp mode does not succeed, the set $U$ is reset to $U(\psi)$, where $\psi$ is the current state. Then the attempt to satisfy the elements in $U$ is invoked again (Line 5). Now if the attempt is successful then $U$ is updated and the assignment found is returned to LTL-CHECK and $getState$ terminates (Line 9,10). Otherwise, it turns into the \up mode. Note the two constraints added in Line 5 enables SAT solvers to prune those states which postpone all elements in $U$ and whose next states are already explored.

\begin{algorithm}
\caption{Implementation of the \sp mode}\label{alg:sp}
\begin{algorithmic}[1]    
\REQUIRE an LTL formula $\phi$ and a global set $U$; 
\ENSURE  a next state of $\phi$; 
       \STATE Let \[ \psi = \left\{
                            \begin{array}{l l}
                                 visisted[0](\phi) & \quad \text{if U = $\emptyset$ for the first time;}\\
                                  & \quad \text{otherwise,}\\
                                  \bigvee_{i\leq pos} visisted[i] &  \quad \text{pos is the previous position}\\
                                        &  \quad \text{when U becomes $\emptyset$.}
                            \end{array} \right.\]
       \STATE Let P be an assignment of \\
       $(\f{\phi}\wedge X(\psi))^p\wedge\bigwedge_{\lambda\in explored}\neg (X(\lambda))^p$;
       \IF {P is not empty}
         \RETURN $P$;
       \ELSE
         \STATE Turn back into the \ue mode.
       \ENDIF
  \end{algorithmic}
\end{algorithm}

To find a satisfying cycle, the \sp mode tries to check whether there is a visited state whose position in $visited_S$ is less or equal than the position where $U$ becomes empty in previous time. Specially if $U$ becomes empty for the first time then only the initial state can be considered. Line 1 of Algorithm \ref{alg:sp} assigns the disjunction of these states to be a constraint $\psi$. Then line 2 shows the inquiry to SAT solvers to get an assignment whose next state appears before $U$ becomes empty in previous time. If the inquiry succeeds, the \sp mode returns the assignment found to LTL-CHECK, and $getState$ terminates (Line 4). Otherwise the \sp mode turns back to the \ue mode for further processing (Line 6). 

Consider the visited state $\psi$ whose position in $visited_S$ is the one before where $U$ becomes empty (If $U$ becomes empty for the first time, then $\psi$ must be the input formula $\phi$). So the cycle formed by the \sp mode succeeds to satisfy all Until formulas in $CF(\psi)$ -- since $U$ is reset to be $U(\psi)$ after $\psi$ and all elements in $U$ are satisfied when $U$ becomes empty again. Thus according to Theorem \ref{thm:satreasoning} a satisfying model is found.

\subsection{Unsat Core Extraction}
It may happen that all elements in $U$ are postponed in the \ue mode. Then the function turns into the \up mode, trying to figure out whether 1). $\psi$ is finally satisfied; or 2). $\psi$ is postponed forever. The trivial way to check all reachable states of $\phi$ postpone $\psi$ is proven not efficient due to its large cost, so we must introduce a more clever methodology. 

Now let's dig into the reason why $\psi$ is postponed in $\phi$. $\psi$ is postponed in $\phi$ iff the formula $\f{\phi}\wedge v(\psi)$ is unsatisfiable. So there must be a minimal unsat core $S_1\subseteq CF(\phi)$ such that 
\begin{itemize}
  \item $\f{\bigwedge S_1}\wedge v(\psi)$ is unsatisfiable; and 
  \item for each $S_1'\subset S_1$, $\f{\bigwedge S_1'}\wedge v(\psi)$ is satisfiable.
\end{itemize}
Note there are already works on computing such minimal unsat core (see \cite{MSL11}), and we can directly apply them here. 

Now the task changes to check whether there exists a next state $\phi_1$ of $\phi$ that can avoid the appearance of $S_1$, i.e. $S_1\not\subseteq CF(\phi_1)$. We can achieve this via SAT solvers by feeding them the formula $\f{\phi}\wedge \neg X(\bigwedge S_1)$. If the formula is satisfiable, then the modeling assignment is the next state that can avoid $S_1$; Otherwise, there must be a minimal unsat core $S_2\subseteq CF(\phi)$ to $\bigwedge S_1$, making $\f{\bigwedge S_2}\wedge \neg X(\bigwedge S_1)$ is unsatisfiable -- as $S_1$ to $\psi$. Then the task changes to check whether the avoidance of $S_2$ can be achieved in the next state of $\phi$ $\ldots$

This is a recursive process and one can see we may maintain a sequence of minimal unsat cores $\rho=S_1,S_2,\ldots$ during computing the next state of $\phi$. Then the question raises up that if there is no other next state other than itself for current state $\phi$, how can it terminate the minimal-unsat-core computation?

Let $\theta_i=\psi\wedge \bigwedge_{1\leq j\leq i} S_j$, i.e. the formula which conjuncts $\psi$ and minimal unsat cores from $S_1$ to $S_i$. Apparently it holds $\theta_{i+1}\Rightarrow \theta_i (i\geq 1)$ and what we expect more for finding a next state is $\theta_{i}\not\Rightarrow \theta_{i+1}$. Once $\theta_{i}\Rightarrow \theta_{i+1}$ holds as well, it indicates $\theta_i$ is the unsat core we want to capture Until formula $\psi$ is postponed forever from $\phi$. Actually the reason is under this case we can prove that, $\f{\theta_i}\wedge\neg X(\theta_i)$ is unsatisfiable which means $\psi$ will be postponed in all reachable states of $\phi$.

Assume a next state $\phi_1$ of $\phi$ is found according to above strategy and a sequence $\rho=S_1,S_2,\ldots,S_k (k\geq 1)$ is maintained. Now $\phi_1$ tries to avoid $S_{k-1}$ in its next state -- Note $S_k$ is not in $CF(\phi_1)$  but $S_{k-1}$ is. (If $k=1$ then $\phi_1$ tires to avoid $\psi$ in its next state). The corresponding formula is $\f{\phi_1}\wedge \neg X(S_{k-1})$. This attempt may not succeed, and there may be another minimal unsat core $S_k'$ (not $S_k$) to $S_{k-1}$. We must also maintain this information in the sequence. So it turns out the sequence we have to maintain is the sequence of set of minimal unsat cores, i.e. $\rho=Q_1,Q_2,\ldots,Q_k$ where each $Q_i$ is a set of minimal unsat cores. 

For example consider $\phi=(a U \neg b)\wedge b\wedge X b\wedge XXb$, we can see easily $\psi=a U \neg b$ is postponed currently and $Q_1=\{\{b\}\}$. Moreover $Q_2=\{\{X b\}\}$ and $Q_3=\{\{XXb\}\}$. According to our strategy above, we first try to avoid elements in $Q_3$, that is to check $\f{\phi}\wedge \neg X(\bigvee_{S\in Q_3}\bigwedge S)$. Then we get the next state $\phi_1=(a U \neg b)\wedge b\wedge Xb$. Similarly we get $\phi_2=(a U \neg b)\wedge b$ and $\phi_3=a U \neg b$. By then we know $\psi=a U \neg b$ is not postponed. 

For the formula $\phi= F a \wedge G \neg a$, we know $\psi=F a$ is postponed currently and $Q_1=\{\{Fa, G\neg a\}\}$. Since we know that $\theta_1=\psi\wedge\bigvee_{S\in Q_1}\bigwedge S= Fa \wedge G\neg a$ and $\theta_1\wedge \neg X\theta_1$ is unsatisfiable -- which means an unsat core is already found and $\psi$ will be postponed from $\phi$ forever. So we can terminate by returning unsatisfiable and the unsat core $Fa\wedge G\neg a$.

For a more complicated example we consider $\phi= F(\neg a\wedge \neg b)\wedge a\wedge G((a\rightarrow Xb)\wedge (b\rightarrow Xa))$. The Until formula $\psi=F(\neg a\wedge \neg b)$ is postponed currently and $Q_1=\{\{a\}\}$. To avoid the elements in $Q_1$ in next state, i.e. $\f{\phi}\wedge \neg X(\bigvee_{S\in Q_1}\bigwedge S)$, we can get the next state $\phi_1=F(a\wedge b)\wedge b\wedge G((a\rightarrow Xb)\wedge (b\rightarrow Xa))$. Now $\psi$ is also postponed due to $b\in CF(\phi_1)$, and we update $Q_1=\{\{a\},\{b\}\}$. Then we collect existed states $\phi,\phi_1$ together into set $Sts$, and try to avoid elements in $Q_1$ in the next state of states -- the formula is $(\bigvee_{\phi\in Sts}\f{\phi})\wedge \neg X(\bigvee_{S\in Q_1}\bigwedge S)$. But this attempt still fails. So we get $Q_2=\{\{a, G((a\rightarrow Xb)\wedge (b\rightarrow Xa))\},\{b, G((a\rightarrow Xb)\wedge (b\rightarrow Xa))\}\}$. As we see $\theta_2=\psi\wedge (\bigvee_{S\in Q_1}\bigwedge S)\wedge (\bigvee_{S\in Q_2}\bigwedge S)$ is an invariant, i.e. $\f{\theta_2}\wedge X\neg \theta_2$ is unsatisfiable, so we know $\theta_2$ is the unsat core and $\phi$ is unsatisfiable. 

Now we start to define the sequence we maintain in the \up mode, which we call \textit{avoidable sequence}.

\begin{definition}[Avoidable Sequence]
  For an LTL formula $\phi$ and $\psi$ is an Until subformula of $\phi$, the avoidable sequence of $\psi$ is a sequence $\rho=Q_0,Q_1,\ldots$, where $Q_i\subseteq 2^{cl(\phi)}$ and 
  \begin{itemize}
    \item $Q_0=\{\{\psi\}\}$;
    \item For $i\geq 0$, let $\theta_i=\bigwedge_{0\leq j\leq i}(\bigvee_{S\in Q_{j}}\bigwedge S)$ then $S'\in Q_{i+1}$ ($S'\subseteq cl(\phi)$) if, 
    \begin{enumerate}
      \item $(\theta_i \wedge \neg X(\theta_i))$ is satisfiable;
      \item $(\bigwedge S') \wedge (\theta_i \wedge \neg X(\theta_i))$ is unsatisfiable;
      \item For each $S''\subset S'$, $(\bigwedge S'') \wedge (\theta_i \wedge \neg X(\theta_i))$ is satisfiable.
    \end{enumerate} 
  \end{itemize}
  Specially, we say $\rho$ is an unavoidable sequence if there is $k>0$ such that $\theta_{k} \Rightarrow \theta_{k+1}$.
\end{definition}

The avoidable sequence is an abstract way to represent set of states that postpone $\psi$. If a state $\phi'$ in the transition system $T_{\phi}$ satisfies $\phi'\Rightarrow \theta_i$, then $\phi'$ is represented by $\theta_i$. Let $S_{\theta_i} (S_{\theta_{i+1}})$ be the set of states represented by $\theta_i (\theta_{i+1})$, it is easy to see $S_{\theta_{i+1}}\subseteq S_{\theta_i}$ due to $\theta_{i+1}\Rightarrow \theta_i$. For example, assume the avoidable sequence $\rho=\{\{a\}\},\{\{b\}\}$, then we know $\theta_1=a$ and $\theta_2=a\wedge b$. Apparently the state $a\wedge \neg b$ can be represented by $\theta_1$, but not by $\theta_2$. Since the set of states in $T_{\phi}$ is finite, so the avoidable sequence must also be finite. 

\begin{lemma}
  For an LTL formula $\phi$ and $\psi$ is an Until subformula of $\phi$, the avoidable sequence of $\psi$ is finite.
\end{lemma}

Specially when $\rho$ is an unavoidable sequence, i.e. $\theta_{k-1} \Rightarrow \theta_k (k>0)$, it means essentially $S_{\theta_i}=S_{\theta_{i+1}}$. We can prove that in this case $\theta_i$ covers all states postpone $\psi$ forever. Before that we introduce the following lemma.

\begin{lemma}\label{lem:invariant}
  For an LTL formula $\phi$ and $\psi$ is an Until subformula of $\phi$, if $\rho=Q_0,Q_1,\ldots,Q_k (k\geq 1)$ is an unavoidable sequence of $\psi$, then it holds that $\f{\theta_k}\wedge \neg X(\theta_k)$ is unsatisfiable. 
\end{lemma}
\begin{proof}
  Since $\rho$ is an unavoidable sequence, so $S_{\theta_{k-1}}=S_{\theta_k}$. Assume $\lambda$ is a state represented by $S_{\theta_{k-1}}$, and it is also in $S_{k}$ if there is $S\in Q_k$ such that $S\subseteq CF(\lambda)$. Now let's recall the meaning of elements in $Q_k$. Since $S\in Q_k$ so it satisfies $\f{\lambda\wedge\bigwedge S}\wedge \neg X(\theta_{k-1})$ is unsatisfiable. This means the reason (minimal unsat core) causing all next states of $\lambda$ are also in $S_{\theta_k}$ are contained by $\lambda$ itself. Thus it indicates all next states of states in $S_{\theta_{k-1}}$ are also in $S_{\theta_k}$, which formally means $\f{\theta_k}\wedge \neg X(\theta_k)$ is unsatisfiable. \qed
\end{proof}

\begin{algorithm}[H]
\caption{Implementation of the \up mode}\label{alg:improvedup}
  \begin{algorithmic}[1]     
  \REQUIRE an LTL formula $\phi$ postpone all Until formulas in $U$; 
  \ENSURE  a finite path satisfies at least one element of $U$ or an unsat core;
       \STATE Get some reachable states from $\phi$ and put them into $Sts$ (including $\phi$);
       \STATE Let $\rho=\{\{U\}\}$ and $pos = 0$;
       \STATE Let $\theta_{pos} = \bigwedge_{0\leq i\leq pos}\bigvee_{S\in \rho[i]}\bigwedge S$;
       \WHILE {true}
         
         \WHILE {$(\f{\bigvee Sts}\wedge\neg X(\theta_{pos}))^p$ is satisfiable}
           \STATE Let $P$ be the assignment and add $X(P)$ to $Sts$;
           \STATE $pos = pos - 1$ and update $\theta_{pos}$ ($pos$ is changed);
           \IF {$pos < 0$}
             \RETURN the finite path leading from $\phi$ to $X(P)$;
           \ENDIF
         \ENDWHILE
         \STATE Add computed set of minimal unsat cores to $\rho[pos+1]$ (if $\rho[pos+1]$ is not existed, then extend it);
         \STATE $pos=pos+1$ and update $\theta_{pos}$;
         \WHILE {$\f{\bigvee Sts}\wedge\neg X(\theta_{pos})$ is unsatisfiable}
           \STATE Add computed set of minimal unsat cores to $\rho[pos+1]$ (if $\rho[pos+1]$ is not existed, then extend it);
           \STATE $pos=pos+1$ and update $\theta_{pos}$;
           \IF {$\f{\theta_{pos}}\wedge\neg X(\theta_{pos})$ is unsatisfiable}
             \RETURN $\theta_{pos}$ as the unsat core;
           \ENDIF
         \ENDWHILE
         \STATE Let $P$ be the assignment of $(\f{\bigvee Sts}\wedge\neg X(\theta_{pos}))^p$ and add $X(P)$ to $Sts$;
         \STATE $pos = pos - 1$ and update $\theta_{pos}$;
         
       \ENDWHILE
  \end{algorithmic}
\end{algorithm}

Let $S_{\rho}$ be the set of states represented by the avoidable sequence $\rho$, then we have

\begin{theorem}
  For an LTL formula $\phi$ and $\psi$ is an Until subformula of $\phi$, if $\rho$ is an unavoidable sequence of $\psi$, then all states represented by $S_{\rho}$ are unsatisfiable. 
\end{theorem}
\begin{proof} 
  From Lemma \ref{lem:invariant} we know all next states of states in $S_{\rho}$ are also in $S_{\rho}$. And since every state represented by $S_{\rho}$ can postpone $\psi$, so all states in $S_{\rho}$ together can postpone $\psi$ forever. Thus all states represented by $S_{\rho}$ are unsatisfiable. \qed
\end{proof}

Now we present the improved algorithm for \up mode. The algorithm maintains the information of avoidable sequence in the mode and utilizes it to locate the result. We should claim that, computing elements of avoidable sequence is relatively expensive so far, and especially for extending the length of the sequence. Consider that if finally the Until formula turns out to be satisfiable, then it may wast time to maintain unnecessary long sequence. To balance these situations, our algorithm starts from a set of states reachable from the initial postponed state rather than only itself, in which case it can increase the possibility to find the Until formula satisfiable earlier.

Note that Let $\rho=Q_0,Q_1,\ldots,Q_k$ and we use $\rho[i]$ to represent $Q_i$ in the algorithm. The variable $pos$ points to the position of $\rho$ in which the elements should be avoided currently. The notation $X(P)$ means the set of Next formulas in $P$ (they form the next state indeed). In Line 1, users can decide by themselves the number of reachable states and how to acquire them. 

\section{More experiments on LTL-Satisfiability Checking}\label{app:detail_results}

\begin{table*}[htb]
\centering
    \caption{Experimental results on the Schuppan-collected benchmark for satisfiable formulas. 
     Each cell lists a tuple $\langle t, n\rangle$ where $n$ is the total number of solved formulas and 
     $t$ is the total checking time for solving these $n$ cases (in seconds).}\label{tab:result3}
    \scalebox{0.8}
    {
    \begin{tabular}{|l|l|r||l|r||l|r||l|r||l|r||l|r||}
    \hline
    Formula type  & \multicolumn{2}{c|}{ls4} & \multicolumn{2}{c|}{TRP++} & \multicolumn{2}{c|}{NuXmv-BMCINC} &  \multicolumn{2}{c|}{Aalta\_v1.2} & \multicolumn{2}{c|}{NuXmv-IC3-Klive} & \multicolumn{2}{c|}{\tool}\\
    \hline
    /acacia/example  &  152 & 49 &  192 & 50 &  0 & 50 & 1 & 50 &  8 & 50 &  1 & 50\\
    \hline
    /acacia/demo-v3  &  748  & 40 &  554  & 34 & 3  & 72 & 3 & 72 & 30  & 72 & 3 & 72\\
    \hline
    /acacia/demo-v22  &  60  & 20 &  67  & 20 &  0 & 20 &  2 & 20 & 4 & 20 &  1 & 20\\
    \hline
    /alaska/lift  &  487 & 22 &  322 & 16 &  282 & 238 & 4084  & 163 &  529 & 233 &  367 & 229\\
    \hline  
    /alaska/szymanski  &  27 & 8 &  43 & 4 &  0 & 8 & 1  & 8 &  2 & 8 &  0 & 8\\
    \hline
    /anzu/amba  &  0 & 0 &  0 & 0 &  116 & 95 &  2686 & 65 & 582 & 94 &  273 & 98\\
    \hline 
    /anzu/genbuf  &  0 & 0 &  0 & 0 &  122 & 109 &  3343 & 79& 570 & 107 &  422 & 116\\
    \hline
    /rozier/counter  &  1214 & 90 &  1851 & 108 &  25 & 88 &  928 & 60 & 88 & 87 &  289 & 114\\
    \hline  
    /rozier/formulas  &  163 & 3890 &  6087 & 3370 &  88 & 3890 &  1372 & 3890 & 649 & 3890 &  28 & 3890\\
    \hline  
    /rozier/pattern  &  936 & 260 &  1230 & 251 &  1025 & 480 & 8  & 488 &  2232 & 471 &  9 & 488\\
    \hline  
    /schuppan/O1formula  &  49 & 10 &  51 & 10 &  7 & 27 & 2 & 27 & 59  & 27 & 1 & 27\\
    \hline  
    /schuppan/O2formula  &  77 & 10 &  89 & 10 &  98 & 24 &  2 & 27 & 253 & 27 &  0 & 27\\
    \hline  
    /schuppan/phltl  &  87 & 5 &  33 & 4 &  135 & 17 &  233 & 10 & 78 & 18 &  1 & 17\\
    \hline  
    /trp/N5x  &  81 & 371 &  83 & 371 &  10 & 371 &  309 & 360 & 145 & 371 &  12 & 371\\
    \hline 
    /trp/N5y  &  334 & 234 &  331 & 234 &  8 & 234 & 16 & 234 & 84 & 234 &  10 & 234\\
    \hline  
    /trp/N12x  &  4425 & 268 &  1639 & 65 &  33 & 625 &  768 & 620 & 531 & 625 &  94 & 625\\
    \hline  
    /trp/N12y  & 3173 & 118 &  3062 & 111 &  29 & 313 &  413 & 313 & 318 & 313 &  12 & 313\\
    \hline  
    /forobots  &  696 & 53 &  914 & 53 &  3 & 53 &  280 & 53 & 27 & 53 &  30 & 53\\
    \hline 
    Total & 12718 & 5448 & 16556 & 4711 & 1994 & 6714 & 14451 & 6689 & 6189 & 6700 & 1554 & 6750\\
    \hline
\end{tabular}
}
\end{table*}

\begin{table*}[htb]
\centering
    \caption{Experimental results on the Schuppan-collected benchmark for unsatisfiable formulas. 
     Each cell lists a tuple $\langle t, n\rangle$ where $n$ is the total number of solved formulas and 
     $t$ is the total checking time for solving these $n$ cases (in seconds).}\label{tab:result4}
    \scalebox{0.8}
    {
    \begin{tabular}{|l|l|r||l|r||l|r||l|r||l|r||l|r||}
    \hline
    Formula type  & \multicolumn{2}{c|}{ls4} & \multicolumn{2}{c|}{TRP++} & \multicolumn{2}{c|}{NuXmv-BMCINC} &  \multicolumn{2}{c|}{Aalta\_v1.2} & \multicolumn{2}{c|}{NuXmv-IC3-Klive} & \multicolumn{2}{c|}{\tool}\\
    \hline
    /acacia/example  &  0 & 0 &  0 & 0 &  0 & 0 & 0 & 0 &  0 & 0 &  0 & 0\\
    \hline
    /acacia/demo-v3  &  0  & 0 &  0  & 0 & 0  & 0 & 0 & 0 & 0  & 0 & 0 & 0\\
    \hline
    /acacia/demo-v22  &  0  & 0 &  0  & 0 &  0 & 0 & 0 & 0 & 0 & 0 &  0 & 0\\
    \hline
    /alaska/lift  &  73 & 22 &  3 & 2 &  77 & 8 & 384  & 10 &  38 & 34 &  544 & 30\\
    \hline  
    /alaska/szymanski  &  0 & 0 &  0 & 0 &  0 & 0 & 0  & 0 &  0 & 0 &  0 & 0\\
    \hline
    /anzu/amba  &  0 & 0 &  0 & 0 &  0 & 0 &  0 & 0 & 0 & 0 &  0 & 0\\
    \hline 
    /anzu/genbuf  &  0 & 0 &  0 & 0 &  0 & 0 &  0 & 0& 0 & 0 &  0 & 0\\
    \hline
    /rozier/counter  &  0 & 0 &  0 & 0 &  0 & 0 &  0 & 0 & 0 & 0 &  0 & 0\\
    \hline  
    /rozier/formulas  &  4 & 110 &  66 & 107 &  29 & 91 &  40 & 100 & 15 & 110 &  1 & 110\\
    \hline  
    /rozier/pattern  &  0 & 0 &  0 & 0 &  0 & 0 & 0  & 0 &  0 & 0 &  0 & 0\\
    \hline  
    /schuppan/O1formula  &  103 & 10 &  27 & 9 &  7 & 27 & 2 & 27 & 36  & 27 & 2 & 27\\
    \hline  
    /schuppan/O2formula  &  106 & 9 &  16 & 3 &  3 & 2 &  2 & 27 & 69 & 20 &  5 & 27\\
    \hline  
    /schuppan/phltl  &  64 & 4 &  19 & 3 &  22 & 4 &  89 & 4 & 15 & 7 &  36 & 4\\
    \hline  
    /trp/N5x  &  62 & 109 &  62 & 109 &  17 & 100 &  139 & 88 & 42 & 109 &  8 & 109\\
    \hline 
    /trp/N5y  &  113 & 46 &  104 & 45 &  0 & 0 & 130 & 10 & 18 & 46 &  16 & 46\\
    \hline  
    /trp/N12x  &  0 & 0 &  0 & 0 &  56 & 117 &  456 & 20 & 174 & 175 &  64 & 174\\
    \hline  
    /trp/N12y  & 277 & 6 &  180 & 4 &  0 & 0 &  34 & 13 & 95 & 67 &  238 & 67\\
    \hline  
    /forobots  &  293 & 25 &  388 & 25 &  2 & 7 &  280 & 10 & 10 & 25 &  32 & 23\\
    \hline 
    Total & 1102 & 323 & 906 & 307 & 215 & 356 & 1556 & 309 & 512 & 620 & 946 & 617\\
    \hline
\end{tabular}
}
\end{table*}

\begin{figure*}
      \centering
      \includegraphics[scale = 0.5]{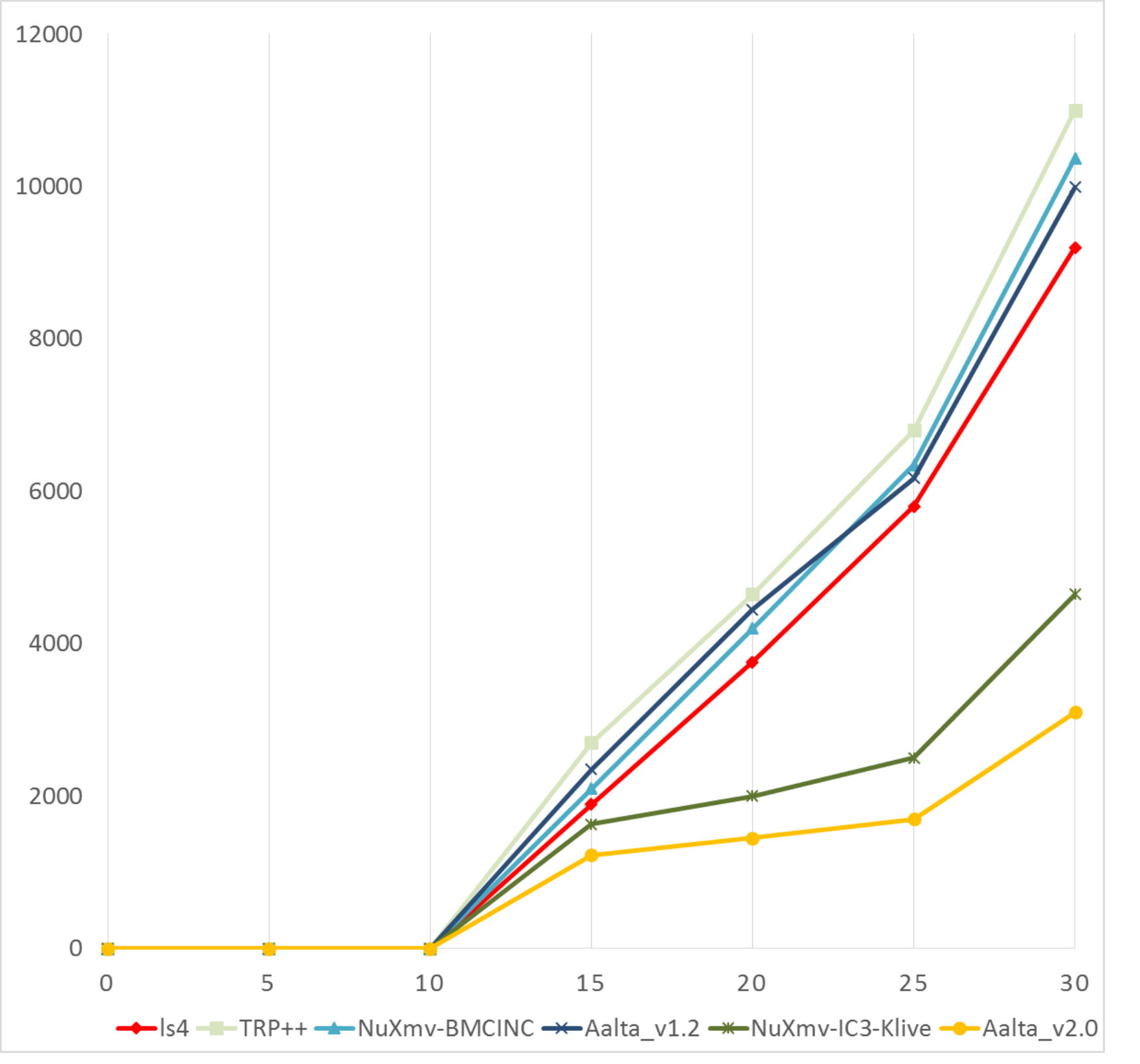}
      \caption{Results for LTL-satisfiability checking on Random Conjunction Formulas.}\label{fig:random}
\end{figure*}

This section shows more experimental results on LTL-satisfiability checking. First we complete the results in Table \ref{tab:result2} and list the results on satisfiable and unsatisfiable formulas separately, which are respectively shown in Table \ref{tab:result3} and Table \ref{tab:result4}. 
Slightly different with Table \ref{tab:result2}, each cell of these two tables lists a tuple $\langle t, n\rangle$ 
where $n$ is the total number of solved formulas and $t$ is the total checking time for solving these $n$ cases (in seconds). 
In these two table \tool is tested by using heuristics. The separation may help readers understand better of checking performance on satisfiability and unsatisfiability.

In additional to the \textit{schuppan-collected} benchmarks, we also tested all 
solvers on the \textit{random conjunction formulas}, which is proposed in \cite{LZPVH13}. 
A random conjunction formula $RC(n)$ has the form of $\bigwedge_{1\leq i\leq n}P_i(v_1,v_2,\ldots,v_k)$, 
where $n$ is the number of conjunctive elements and $P_i(1\leq i\leq n)$ is a randomly chosen pattern 
formula used frequently in practice\footnote{http://patterns.projects.cis.ksu.edu/documentation/patterns/ltl.shtml}. 
The motivation is that typical temporal assertions may be quite small in practice. And what makes 
the LTL satisfiability problem often hard is that we need to check 
\emph{large collections of small temporal formulas}, so we need to
check that the conjunction of all input assertions is satisfiable. 
In our experiment, the number of $n$ varies from 1 to 30, 
and for each $n$ a set of 100 conjunctions formulas are randomly chosen. 
The experimental results are shown in Fig. \ref{fig:random}. It shows that \tool 
(with heuristic) performs best among tested solvers, and comparing to the second best solver (NuXmv),
 it achieves approximately the 30\% speed-up. 

\section{SMT-based Temporal Reasoning}\label{sec:demonstration}\label{app:smt}
An additional motivation to base explicit temporal reasoning on SAT solving is the need to handle
LTL formulas with \emph{assertional atoms}, that is, atoms that are non-boolean state assertions, 
e.g., assertions about program variables, such as $k\leq 10$. Existing explicit temporal-reasoning 
techniques abstract such assertions as propositional atoms. Consider, for example, the LTL formula 
$\phi=\bigwedge_{1\leq i\leq n}F(k=i)$, which asserts that $k$ should assume all values 
between $1$ and $n$. By abstracting $k=i$ as $p_i$, we get the formula 
$\phi'=\bigwedge_{1\leq i\leq n}Fp_i$, but the transition system for the abstract
formula has $2^n$ states, while the transition system for the original formula has only 
$n$ states.  This problem was noted, but not solved in \cite{TRV12}, but it is obvious that 
reasoning about non-Boolean assertions requires reasoning at the assertion level. 
Basing explicit temporal reasoning on SAT solving, would enable us to lift it to 
the assertion level by using Satisfiability Modulo Theories (SMT) solving. 
SMT solving is a decision problem for logical formulas in combinations of background
theories expressed in classical first-order logic.  Examples of theories 
typically used  are the theory of real numbers, the theory of integers, 
and the theories of various data structures such as lists, arrays, bit vectors, 
and others.  SMT solvers have shown dramatic progress over the past couple of decades 
and are now routinely used in industrial software development \cite{DB08}.  

So far, we described how to use SAT solving for checking satisfiability of \emph{propositional
LTL formulas}. And in this section we show that our approach can be extended to reason assertional 
LTL formulas.  In many applications, we need to handle LTL formulas with \emph{assertional atoms}, 
that is, atoms that are non-boolean state assertions, e.g., assertions about program variables.
For example, Spin model checker uses temporal properties expressed in LTL using assertions about 
Promela state variables \cite{Hol03}.  Existing explicit temporal-reasoning tools, e.g., 
SPOT \cite{DP04}, abstract such assertions as propositional atoms. 

Recall that we utilize SAT solvers in our approach to compute assignments of formulas $\phi^p$ (with $\phi$ is in XNF). 
The states of transition system are then obtained from these assignments. When $\phi$ is an assertional LTL formula,
the formula $\phi^p$ is not a propositional formula, but a Boolean combination of theory atoms, for an appropriate theory.
Thus, our approach is still applicable, except that we need to replace the underlying SAT solver by an SMT solver.

Consider, for example the formula $\phi=(F (k=1)\wedge F (k=2))$. The XNF of $\phi$, i.e. $\f{\phi}$, is 
$((v(F(k=1))\wedge(k=1))\vee (\neg v(F(k=1))\wedge XF(k=1)))\wedge ((v(F(k=2))\wedge(k=2))\vee (\neg v(F(k=2))\wedge XF(k=2)))$. 
If we use a SAT solver, we can obtain an assignment such as 
$A=\{(k=1), v(F(k=1)),\neg XF(k=1),(k=2), v(F(k=2)),\neg XF(k=2)\}$, which is consistent
propositionally, but inconsistent theory-wise. This can be avoided by using an SMT solver.
Generally for a formula $\phi_n=\bigwedge_{1\leq i\leq n}F(k=i)$, 
there are $O(2^n)$ states generated in the transition system by the SAT-based approach, 
but only $n$ states need to be generated. This can be achieved by replacing the SAT solver 
in our approach by an SMT solvers. The performance gap between the SAT-based approach and the SMT-based
approach would be exponential. Indeed, SPOT performance on the formulas $\phi_n$ is exponential in $n$.

As proof of concept, we checked satisfiability of the formulas $\phi_n$, for $n=1,\ldots,100$, 
by Aalta\_v2.0.  We then replaced Minisat by Z3, a state-of-the-art SMT solver \cite{DB08}.
The performance results show indeed an exponential gap between the SAT-based approach and the SMT-based approach, 
which is shown in Fig. \ref{fig:demonstration}.  
(Of course, we also gain in correctness: the formula $F(k=1 \wedge k=2)$ is satisfiable
when considered propositionally, but unsatisfiable when considered assertionally.)
Applying SMT-based techniques in other temporal-reasoning tasks, such as
translating LTL to B\"uchi automata \cite{GPVW95} or to runtime monitors~\cite{TRV12},
is a promising research direction.

\begin{figure}
      \centering
      \includegraphics[scale = 0.6]{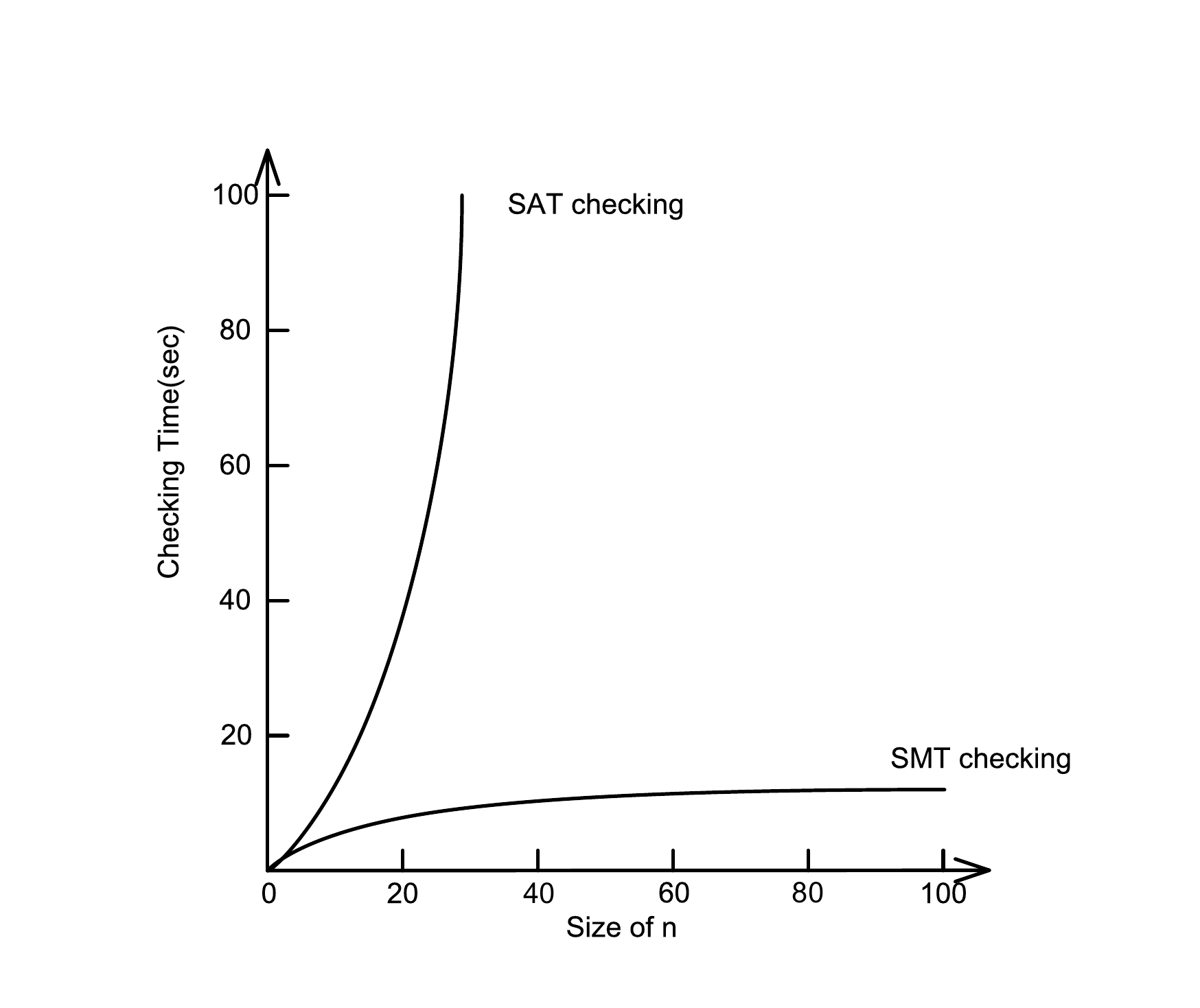}
      \caption{Results for LTL-satisfiability checking on $\bigwedge_{1\leq i\leq n}F(k=i)$.}\label{fig:demonstration}
\end{figure}

\end{document}